\documentclass{article}[a4paper,11pt]

\usepackage[a4paper, margin=2.5cm]{geometry}
\usepackage{graphicx} 
\usepackage{url}
\usepackage{xcolor}
\usepackage{amsmath}
\usepackage{amsfonts}
\usepackage{enumitem}
\usepackage{amsthm}
\usepackage{authblk}
\usepackage{booktabs}

\usepackage[section]{placeins} 
\usepackage{subcaption} 
\allowdisplaybreaks
\newtheorem{theorem}{Theorem}[section]

\newtheorem{lem}[theorem]{Lemma}
\newtheorem{rem}[theorem]{Remark}
\newtheorem{thm}[theorem]{Theorem}
\newtheorem{cor}[theorem]{Corollary}
\newtheorem{pro}[theorem]{Proposition}
\newtheorem{ex}{Example}[section]

\usepackage{tikzit} 

\tikzstyle{new style 0}=[fill=black, draw=black, shape=circle]

\tikzstyle{new edge style 0}=[-, draw={rgb,255: red,108; green,0; blue,108}]
\tikzstyle{new edge style 1}=[-, draw={rgb,255: red,173; green,0; blue,173}]
\tikzstyle{new edge style 2}=[-, draw={rgb,255: red,255; green,128; blue,0}]
\tikzstyle{new edge style 3}=[-, draw={rgb,255: red,74; green,1; blue,158}]
\tikzstyle{new edge style 4}=[-, draw={rgb,255: red,128; green,128; blue,128}]
\tikzstyle{new edge style 5}=[-, draw={rgb,255: red,64; green,64; blue,64}]
\tikzstyle{new edge style 6}=[-, draw={rgb,255: red,191; green,191; blue,191}]
\tikzstyle{new edge style 7}=[-, draw={rgb,255: red,191; green,255; blue,0}]
\tikzstyle{new edge style 8}=[-, draw=cyan]

\usepackage[doi=false, url=false]{biblatex}
\usepackage{dsfont} 

\definecolor{violet}{RGB}{160, 0, 160}

\usepackage[T1]{fontenc}
\usepackage{helvet}

\usepackage{stmaryrd}

\NewDocumentEnvironment{alignb}{b}{%
  \begin{align*}
  \refstepcounter{equation} #1 \tag{\theequation}
  \end{align*}
}{\ignorespacesafterend}

\newcommand{\dd}{\:\mathrm{d}}

\newcommand{\mN}{\mathbb{N}}

\newcommand{\mR}{\mathbb{R}}

\newcommand{\tin}{$t\in \mR_{\geq 0}$ }
\newcommand{\tinb}{$t\in \mR_{\geq 0}$, }

\newcommand{\ain}{$a \in \{1,\dots,d\}$ }
\newcommand{\DD}{\{1,\dots,d\}}
\newcommand{\mI}{\mathcal{I}}

\newcommand{\inti}{\int\limits_{\mI^d}}
\newcommand{\intii}{\inti\inti}

\newcommand{\intia}{\int\limits_{\mathcal{Q}}}
\newcommand{\intiai}{\intia\inti}
\newcommand{\intiaia}{\intia\intia}

\newcommand{\xya}{\phi_{xy\alpha}}
\newcommand{\yxe}{\phi_{yx\eta}}

\title{Multi-Dimensional Opinion Formation}
\author[1]{Hanna Bartel}
\author[2,1]{Martin Burger}
\author[3]{Marie-Therese Wolfram}
\affil[1]{Department of Mathematics, University of Hamburg, Bundesstr. 55, 20146 Hamburg, Germany.}
\affil[2]{Helmholtz Imaging, Deutsches Elektronen-Synchroton DESY, Notkestr. 85, 22607 Hamburg, Germany}
\affil[3]{Warwick Mathematics Institute, University of Warwick, Gibbet Hill Road, CV47AL Coventry, UK}
\date{}

\begin{document}
\maketitle

\begin{abstract}
In this paper we propose and investigate a multi-dimensional opinion dynamics model where people are characterised by both opinions and importance weights across these opinions. Opinion changes occur through binary interactions, with a novel coupling mechanism: the change in one topic depends on the weighted similarity across the full opinion vector. We state the kinetic equation for this process and derive its mean-field partial differential equation to describe the overall dynamics. Analytical computations and numerical simulations confirm that this model exhibits a variety of qualitatively distinct stationary states, and we demonstrate that the final opinion structures are critically determined by the people's opinion weights.
\end{abstract}

\section{Introduction}

There has been extensive research on opinion formation models in different scientific disciplines in the last decades. Most models focus on the dynamics of a single topic, and assume that opinions change through binary interactions with like-minded people. In this paper we propose a new mathematical model to describe the evolution of people discussing and changing their opinion on multiple related topics, for example considering the evolution of people's opinions on climate change, sustainable energy and vegetarianism, thereby providing a more realistic representation of opinion dynamics. In the proposed model, the change in one opinion depends on the closeness in all opinions as well as the individual rating of their importance. The proposed dynamics lead to the formation of ange{various stationary states beyond simple consensus, which we will investigate using analytical and computational tools. 

Classical opinion formation models mostly focus on the evolution of a person's opinion, modelled by a continuous variable on a bounded interval, which changes due to interactions with others. In consensus formation, people average their opinion with others sufficiently close - this closeness can be measured in terms of the opinion distance, and possibly modulated by an underlying social network. The most famous works on consensus formation models  include the contributions of \citeauthor{hegselmannkrause} \cite{hegselmannkrause}, \citeauthor{deffuant2000mixing} \cite{deffuant2000mixing} and \citeauthor{degroot1974reaching} \cite{degroot1974reaching}.  
In the last decades, methods from statistical mechanics - in particular kinetic theory - have been proposed to analyse the overall dynamics of large interacting populations. These contributions go back to the seminal work of \citeauthor{toscani_kinetic_2006}, see \cite{toscani_kinetic_2006}, who first analysed the respective kinetic equations for the population distribution in suitable scaling limits. Boudin et al \cite{boudin_kinetic_2009} proposed a kinetic model for multi-dimensional opinion formation in the context of elections, each opinion corresponding to the support of a specific party.  Various generalisations and extensions of his ideas have been proposed and investigated in the literature, studying for example the impact of leaders \cite{albi_opinion_2016, during_boltzmann_2009}, underlying network structures \cite{bayraktar2023graphon,during2024breaking, nugent_evolving_2023} or exogenous shocks \cite{bondesan2024kinetic}. 

Multi-dimensional models for opinion formation received far less attention in research. So far, generalisations of the Hegselmann-Krause model \cite{hegselmannkrause} for multiple opinions have been studied in \cite{fortunato_vector_2005} and \cite{nedic_multidhk}. In these papers people interact if their opinions are sufficiently close and in case of an interaction all opinions are updated. Similarly, in \cite{pedraza_analytical_2021} \citeauthor{pedraza_analytical_2021} proposed a multi-dimensional model in which people only interact on some of the topics. None of these generalisations consider a weighting across opinions as proposed in this paper. 
An extension of DeGroot's model for consensus formation to the multi-dimensional setting was proposed and analysed in \cite{noipitak_dynamics_2021}. More general multi-dimensional opinion formation models, which include for example the effects of social networks or account for cognitive dissonance theory (which postulates that people do not have contradictory opinions on different topics), were considered in \cite{ojer_social_2025, rodriguez_collective_2016,schweighofer_agent-based_2020, schweighofer_weighted_2020, noorazar_classical_2020}. Solutions to these models exhibit complex dynamics, such as polarisation and ideology alignment.
We note that a different multi-dimensional opinion formation model was published around the same time as this work, see \cite{doi:10.1137/25M1730818}. Here the authors start with the discrete in time version of bounded confidence models and base their interactions on a so-called discordance function. This discordance function corresponds to the distance function with equal opinion weights considered in this work. Due to the time discrete nature of this recent work, the difference in the distance function, as well as the study in a distinct setting (multi-dimensional opinion formation on networks), the two works can be understood as complementary studying different scales and levels of generality of the same underlying phenomena.

In this paper we propose and investigate a novel model for multi-dimensional opinion formation. Our main contributions can be summarised as follows:
\begin{enumerate}[nosep]
    \item Formulation of a multi-dimensional opinion formation model, which accounts for individual rating of importance (of a specific topic).
    \item Analysis of the respective mean-field model and first insights on the structure of stationary states.
    \item Confirmation (analytical and computational) that the proposed model leads to complex and more realistic stationary states.
\end{enumerate}
\vspace*{1em}
We start by presenting the underlying microscopic interaction rules and the respective kinetic model in Section \ref{sec:derivation}. Then we discuss existence and properties of solutions to the respective mean-field model in Section \ref{sec:analysis}. Section \ref{sec:statsol} focuses on the structure of stationary states. In Section \ref{sec:numerics} we illustrate the complex dynamics as well as stationary states with computational experiments and we conclude in Section \ref{sec:conclusion}.

\section{A kinetic model for multi-dimensional opinion formation}\label{sec:derivation}

In this section we follow the methodologies introduced in \cite{toscani_kinetic_2006, during_boltzmann_2009, pedraza_analytical_2021} to model the evolution of opinions in large interacting agent systems. The proposed model is based on the following assumptions: 
\begin{itemize}[nosep]
    \item People do not lie, and they know everyone else's current opinions.
    \item Topics are related indirectly via a distance at which people perceive each other, in particular a change in opinion in one topic does not trigger a change in opinion on any other topic.
    \item No exogenous factors are included (such as media or underlying social network structures).
\end{itemize}

For simplicity, we assume that people always discuss every topic in every interaction. We assume further that people are characterised by their opinions $x\in\mR^d$ with $d \in\mN$ and their respective importance weights are $\alpha \in \mathcal{A}:=\{\alpha\in[0,1]^d\mid \sum_{a=1}^d \alpha_a=1\}$. Moreover, the parameter $\beta\in[0,1]$ weighs the importance of the currently considered opinion against the importance of the other opinions.  We define the distance on the $a$-th topic for two opinion vectors $x$ and $y$ in $\mR^d$ as 
\begin{alignb}\label{eq:p_distance}
    p_a(x,y,\alpha):= \beta |x_a-y_a| + (1-\beta)\sum_{b=1}^d\alpha_b |x_b-y_b|.
\end{alignb}
Note that \eqref{eq:p_distance} is not a norm, since $p_a(x,y,\alpha) = 0$ does not imply that people share the same opinions. We assume that binary interactions between people can be described by an interaction function $\phi:[0,2] \mapsto [0,1]$, which depends on their distance in opinion as defined in \eqref{eq:p_distance}. The function $\phi$ is assumed to be non-increasing, accounting for the fact that people with similar opinions influence each other more than people further apart (a standard assumption in bounded confidence models).\\

We start by defining the binary interaction between two people with opinions and weights $(x,\alpha),(y,\eta)\in\mathcal{Q} := \mR^d \times \mathcal{A}$ and denote their post-interaction opinions by $x^*$ and $y^*$ respectively. They are given by
\begin{alignb}\label{eq:interactions}
    x^* &= x + \gamma \phi_{xy\alpha} \odot (y-x) \\
y^* &= y + \gamma \phi_{yx\eta} \odot (x-y). 
\end{alignb}
The parameter $\gamma\in (0,1)$ describes how strong interactions influence opinions, and $\odot$ denotes component-wise vector multiplication. The function $\phi_{xy\alpha}$ corresponds to the component-wise evaluation of the interaction function $\phi$, i.e.
\begin{align*}
    \phi_{xy\alpha} := \left(
\begin{array}{c}
\phi\left(p_1(x,y,\alpha)\right)\\
\phi\left(p_2(x,y,\alpha)\right)\\
\vdots\\
\phi\left(p_d(x,y,\alpha)\right)\\
\end{array}
\right).
\end{align*}
Note that $p_a(x,y,\alpha)=p_a(y,x,\alpha)$ and thus $\phi_{xy\alpha}=\phi_{yx\alpha}$. However, in general the interaction is not reciprocal due to the difference in $\alpha$ and $\eta$. In particular, this is a difference to $1$D models.\\
\begin{rem}\label{rem:x*y*inI}
    Let $\mI\subset \mR$ be some bounded domain and let $x,y\in\mI^d$. Clearly, for $x^*,y^*$ obtained via \eqref{eq:interactions}, it holds component-wise that $\min(x^*,y^*)\geq \min(x,y)$ and $\max(x^*,y^*)\leq\max(x,y)$, and moreover, $x^*,y^*\in\mI^d$. 
\end{rem}

Consider the distribution function $f=f(x,\alpha,t)$, which describes the ratio of people having opinions $x\in\mR^d$ and importance weights $\alpha\in\mathcal{A}$ at time $t\in \mR_{\geq 0}$. To derive the corresponding mean-field model, we consider
\begin{alignb}\label{eq:dtgl}
\frac{\partial f}{\partial t} = 
G(f,f) - L(f,f),
\end{alignb}
where $G$ and $L$ are the gain and loss term respectively. 
Let now $\kappa$ denote the interaction rate.

The gain term accounts for people changing their opinion to $x^*$ with pre-interaction opinion and weights $(x,\alpha)$. In particular
\begin{align*}
    G(f,f)(x', \alpha,t) =\kappa \intiai \delta\left(x' - x^*(x,y,\alpha)\right)\mathds{1}_{\lbrace x\neq x^{*}\rbrace} f(x,\alpha,t) f(y, \eta, t) \dd x \dd (y,\eta),
\end{align*}
where $x^*$ is given by \eqref{eq:interactions}. Similarly we define the loss, due to interactions of individuals with pre-interaction opinion and weight $(x,\alpha)$ with others having opinion and weights $(y,\eta)$ as 
\begin{align*}
    L(f,f) (x', \alpha,t) = \kappa \intiai \delta\left(x' - x\right) \mathds{1}_{\lbrace x\neq x^*\rbrace }f(x,\alpha,t) f(y, \eta, t) \dd x \dd (y,\eta).
\end{align*}
We see that the right hand side can be written as a collision operator in weak form. Let $\xi(x',\alpha)$ be a suitable test function, then 
\begin{align*}
\intia \Bigl(G(f,f) - L(f,f)\Bigr) \xi(x', \alpha)\dd(x', \alpha) = \kappa\intiaia\Bigl(\xi(x^*, \alpha)-\xi(x,\alpha)\Bigr) f(x,\alpha,t) f(y, \eta, t) \dd (x,\alpha)\dd(y, \eta).
\end{align*}
This corresponds to the well-known weak formulation of a collision operator, see \cite{toscani_kinetic_2006}.

\medskip

Without loss of generality, we assume that $\kappa=1$. To derive the respective mean-field equation we write \eqref{eq:dtgl} in weak form. Let $\xi$ be a test function in $C_c^\infty(\mathcal{Q})$, then 
\begin{alignb}\label{eq:derivation_withtaylor}
    \frac{\dd}{\dd t}\intia & \xi(x,\alpha) f(x,\alpha,t) \dd (x,\alpha) = \intia \xi(x,\alpha)G(x,\alpha)) \dd (x,\alpha) - \intia \xi(x,\alpha) L(x,\alpha) \dd (x,\alpha)\\
    =&\intiaia\Bigl(\xi(x^*, \alpha)-\xi(x,\alpha)\Bigr) f(x,\alpha,t) f(y, \eta, t) \dd (x,\alpha)\dd(y, \eta)\\
    =& \intiaia \Bigl(\xi(x,\alpha) +\gamma\nabla_x \xi(x,\alpha) \cdot (\phi_{xy\alpha}\odot (y-x)) + \mathcal{O}(\gamma^2)-\xi(x,\alpha)\Bigr) f(x,\alpha,t)f(y,\eta,t) \dd (y,\eta) \dd (x,\alpha)\\
    =& \intiaia \Bigl(\gamma\nabla_x \xi(x,\alpha) \cdot (\phi_{xy\alpha}\odot (y-x)) + \mathcal{O}(\gamma^2)\Bigr) f(x,\alpha,t)f(y,\eta,t) \dd (y,\eta) \dd (x,\alpha).
\end{alignb}
Next, we want to compute the grazing collision limit, see \cite{toscani_kinetic_2006} for more details. In doing so we rescale time $\tau = \gamma t$ and define $g(x,\alpha,\tau)=f(x,\alpha,t)$. Then we obtain
\begin{align*}
    \frac{\dd}{\dd \tau}\intia &\xi(x,\alpha) g(x,\alpha,\tau) \dd (x,\alpha)\\
    =& \frac{1}{\gamma}\intiaia \Bigl(\gamma\nabla_x \xi(x,\alpha) \cdot (\phi_{xy\alpha}\odot (y-x)) + \mathcal{O}(\gamma^2)\Bigr) g(x,\alpha,\tau)g(y,\eta,\tau)y,\eta),\tau) \dd (y,\eta) \dd (x,\alpha)\\
    =&\intiaia \Bigl(\nabla_x \xi(x,\alpha) \cdot (\phi_{xy\alpha}\odot (y-x)) + \mathcal{O}(\gamma)\Bigr) g(x,\alpha,\tau)g(y,\eta,\tau)y,\eta),\tau) \dd (y,\eta) \dd (x,\alpha).
\end{align*}
In the limit $\gamma \rightarrow 0$ we obtain that
\begin{align*}
   \frac{\dd}{\dd \tau}\intia &\xi(x,\alpha) g(x,\alpha,\tau) \dd (x,\alpha) =    \intiaia \nabla_x \xi(x,\alpha) \cdot (\phi_{xy\alpha}\odot (y-x)) g(x,\alpha,\tau)g(y,\eta,\tau) \dd (y,\eta) \dd (x,\alpha).
\end{align*}

Returning to the previous notation, using the letter $f$ instead of $g$ and $t$ instead of $\tau$, gives
\begin{alignb}\label{eq:pdeweak}
   &\frac{\dd}{\dd t}\intia \xi(x,\alpha) f(x,\alpha,t) \dd (x,\alpha)
   =&\intiaia \nabla_x \xi(x,\alpha) \cdot (\phi_{xy\alpha}\odot (y-x)) f(x,\alpha,t)f(y,\eta,t) \dd (y,\eta) \dd (x,\alpha),
\end{alignb}
and the corresponding strong form - a Vlasov type equation:
\begin{alignb}\label{eq:pde} 
    \frac{\partial}{\partial t} f(x,\alpha,t) = - \nabla_x \cdot \left(\left(\intia  \phi_{xy\alpha}\odot (y-x) f(y,\eta,t)\dd (y,\eta)\right) f(x,\alpha,t)\right).
\end{alignb}
Notice that since we consider the equation on $\mR^d$, no boundary terms arise in the derivation of the strong form of the PDE.

\begin{rem}
Throughout this paper we will sometimes consider the special case of all people having the same opinion weights. For clarity, in that case we will use the notation $\rho(x,t)$ instead of $f(x,\alpha,t)$. When all people have the same opinion weights, equation \eqref{eq:pdeweak} simplifies to
   \begin{alignb}\label{eq:pdeweak_fixalp}
    \frac{\dd}{\dd t}\inti \xi(x) \rho(x,t) \dd x
    =&\intii \nabla_x \xi(x) \cdot (\phi_{xy}\odot (y-x)) \rho(x,t)\rho(y,t) \dd y \dd x
\end{alignb}
in the weak formulation, and
   \begin{align*}
    \frac{\partial}{\partial t} \rho(x,t) = - \nabla_x \cdot \left(\left(\inti  \phi_{xy}\odot (y-x) \rho(y,t)\dd y\right) \rho(x,t)\right)
\end{align*} 
in the strong formulation.
\end{rem}

\begin{rem}
In \eqref{eq:interactions}, for simplicity, we assume that all opinions of a person change in an interaction. This is not necessarily a realistic assumption, as people often discuss a single topic only. To account for this one can consider the following modification of \eqref{eq:interactions}. Let $\nu$ be a uniformly distributed random variable that takes values in $\DD$. 
\begin{align}
\begin{split}
\label{eq:singleopinteraction}
    x_{\nu}^* &= x_{\nu} + \gamma \phi_{xy\alpha}  \odot (y-x) \\
    y_{\nu}^* &= y_{\nu} + \gamma \phi_{yx\eta} \odot (x-y)\\
    x_{\zeta}^* &= x_{\zeta},\, y_{\zeta}^* = y_{\zeta} \qquad \qquad \qquad \text{for all } \zeta \in \DD \backslash \lbrace \nu \rbrace.
    \end{split}
\end{align}
In \eqref{eq:singleopinteraction} a single opinion is randomly selected and people change their opinion in this component only, leaving the others unchanged. In the mean-field limit \eqref{eq:singleopinteraction} leads to a rescaling in time, in particular
\begin{align*} 
    \frac{\partial}{\partial t} f(x,\alpha,t) = - \frac{1}{d}\nabla_x \cdot \left(\left(\intia  \phi_{xy\alpha}\odot (y-x) f(y,\eta,t)\dd (y,\eta)\right) f(x,\alpha,t)\right).
\end{align*}
The full derivation for this modification can be found in the appendix in Section \ref{app:onlyonetopicperinter}.
\end{rem}

\subsection{Characteristics at boundary}\label{sec:boundary_behabvior}
In the following we would like to consider \eqref{eq:pde} on the bounded domain $\mI^d:=[-1,1]^d$, but we did the derivation on $\mR^d$. We have already seen in Remark \ref{rem:x*y*inI}, that post-interaction opinions are bounded by the respective pre-interaction opinions. Next we will show that if the support of $f$ is on a bounded domain at time $t$, it remains within this bounded domain for all times $t^*>t$.
For this, we assume that $\text{supp}(f(.,t))\subseteq \mI^d$. Let $(x,\alpha)$ be a boundary point of the hypercube $\mI^d$. Then there exist two sets, $\mathcal{B}_+,\mathcal{B}_-\subseteq \{1,\dots,d\}$, $\mathcal{B}_+\cap \mathcal{B}_-=\emptyset$ and $\left(\mathcal{B}_+\neq \emptyset \vee \mathcal{B}_-\neq \emptyset\right)$, such that 
\begin{align*}
    x_a &= \begin{cases}
        1 \quad & \text{ for all } a \in \mathcal{B}_+\\
        -1 \quad & \text{ for all } a \in \mathcal{B}_-.
    \end{cases}
\end{align*}
We define the outer unit normal vector at $x$ as
\begin{align*}
    n_a = \begin{cases}
\frac{1}{\lvert n \rvert} \quad & \text{ for all } a \in \mathcal{B}_+\\
-\frac{1}{\lvert n \rvert} \quad & \text{ for all } a \in \mathcal{B}_-\\
0  \quad & \text{ for all } a\in\{1,\dots,d\}\backslash(\mathcal{B}_+\cup \mathcal{B}_-).
    \end{cases}
\end{align*}
Note that we choose one possible outward normal vector at corners of $\mI^d$. The following computation shows that the characteristics point inwards, i.e.
\begin{align*}
        &\Bigl(\int_{\mR^d \times \mathcal{A}}  \phi_{xy\alpha}\odot (y-x) f(y,\eta,t\dd (y,\eta)\Bigr)\cdot n=\left(\int_{\mI^d \times \mathcal{A}}  \phi_{xy\alpha}\odot (y-x) f(y,\eta,t)\dd (y,\eta)\right)\cdot n\\
        &\qquad=\sum_{a=1}^d \int_{\mR^d \times \mathcal{A}}  \phi_{xy\alpha_a}(y_a-x_a) f(y,\eta,t)\dd (y,\eta) n_a\\
        &\qquad=\sum_{a\in\mathcal{B}_+} \int_{\mR^d \times \mathcal{A}}  \phi_{xy\alpha_a}(y_a-x_a) f(y,\eta,t)\dd (y,\eta) + \sum_{a\in\mathcal{B}_-} \int_{\mR^d \times \mathcal{A}}  \phi_{xy\alpha_a}(y_a-x_a) f(y,\eta,t)\dd (y,\eta)  \\ 
        &\qquad + \sum_{a\in\{1,\dots,d\}\backslash(\mathcal{B}_+\cap \mathcal{B}_-)} \int_{\mR^d \times \mathcal{A}}  \phi_{xy\alpha_a}(y_a-x_a) f(y,\eta,t)\dd (y,\eta)\\
        &\qquad \leq 0,
\end{align*}
where we used the definition of $n_{\alpha}$ in the last inequality.
Hence, the support of $f$ remains in $\mI^d$, and we do not have to impose a boundary condition on \eqref{eq:pde} when considering it on $\mI^d$. Thus, from now on we investigate our model on $\mI^d$ and we will redefine $\Omega$ as $\Omega:=\mI^d\times\mathcal{A}$ for the space of opinion vectors and opinion weights.

\section{Global in time existence and properties of solutions}\label{sec:analysis}
In this section we discuss existence of solutions to \eqref{eq:pde} and their properties. We use the notations $\mathcal{M}(\mathcal{Q})$ and $\mathcal{P}(\mathcal{Q})$ for the spaces of Radon measures respectively probability measures on $\mathcal{Q}$. With slight abuse of notation we still write $f(x,\alpha,t)$ meaning a measure in the first two quantities.

\subsection{Global in time existence}

We use the Picard Lindel\"of Theorem to show existence of solutions to \eqref{eq:pdeweak}. Before doing so, we make the following assumption;
\begin{enumerate}[label={\textbf{(A\arabic*)}}]
    \item\label{it:lip} $\phi:[0,2]\to[0,1]$ is Lipschitz continuous with Lipschitz constant $L\in \mR_{\geq 0}$.
\end{enumerate}
The classic interaction function in bounded confidence models, introduced in \cite{hegselmannkrause}, is $\phi(s) = \mathds{1}_{s \leq R}(s)$ for a given $R \in \mathbb{R}_{\geq 0}$. This function is however not Lipschitz continuous, therefore violating 
Assumption \ref{it:lip}. We can consider a smoothed version, first suggested in \cite{AN2025}, of the following form

    \begin{alignb}\label{eq:phismooth}
    \phi\left(r\right)=\begin{cases} 
    1 & \text{if } r\leq r_1 \\
    q\left(\frac{r_{2}-r}{r_{2}-r_{1}}\right) & \text{if } r_1<r<r_2 \\
    0 & \text{if } r_2\leq r,\end{cases}\end{alignb} 
    
    for $q\left(s\right)=\frac{s^{2}}{s^{2}+\left(1-s\right)^{2}}$ and some $r_1, r_2\in(0,2)$ with $r_1<r_2$. A straightforward calculation shows that $\phi$ as defined in \eqref{eq:phismooth} is indeed Lipschitz continuous.

Now let us consider the characteristic curve in opinion space denoted by $X:\mI^d\times \mathcal{A}\times R_{\geq 0} \to \mI^d\times \mathcal{A}, (x_0,\eta, t) \to (x,\alpha)$, which describes the opinion vector of people with initial opinion vector $x_0$ and opinion weightings $\alpha$ at time $t$. By considering the derivative of $X$ along a characteristic, we obtain for $T\in \mR_{>0}$ and $t\in[0,T]$

\begin{alignb}\label{eq:odeX}
        \frac{\partial}{\partial t} X(x_0,\alpha,t) =& \left( \intia  \phi_{X(x_0,\alpha,t)_x y X(x_0,\alpha)_\alpha}\odot (y-X(x_0,\alpha,t)_x) \dd f(y,\eta,t),0  \right)\\
        =& \left( \intia  \phi_{X(x_0,\alpha,t)_x X(y_0,\eta,t)_x X(x_0,\alpha,t)_\alpha}\odot (X(y_0,\eta,t)_x-X(x_0,\alpha,t)_x) \dd f_0(y_0,\eta),0\right) \\
        =&: u(X(x_0,\alpha,t),\alpha,t),
\end{alignb}
where $f_0$ is the push forward measure of $f$ by $X^{-1}$ (the latter being well-defined by the solution of the ODEs backward in time).

Notice that $u:\mI^d\times \mathcal{A} \to \mR$ is continuous in $t$ (since it only depends on it via the differentiable function $X$) and Lipschitz continuous in $(x,\alpha)$ if $\phi$ is Lipschitz continuous. The following lemma formalizes this observation, the proof follows by direct estimates.
\begin{lem}\label{lem:u_lipcont}
    Let $f \in L^\infty(0,T;\mathcal{M}(\mathcal{Q}))$, and let \ref{it:lip} hold. Then, \begin{align*}
        &u:\mI^d\times \mathcal{A} \times \mR_{\geq 0} \to \mR,\\ &(x,\alpha,t)\to \intia  \phi_{xy\alpha}\odot (y-x)  \dd f(y,\eta,t)
    \end{align*} is Lipschitz continuous with respect to the $l_1$-norm in $(x,\alpha)$, with a uniform Lipschitz constant in time. In particular the weak divergence $\nabla \cdot u$ is bounded almost everywhere.
\end{lem}

We aim to prove a standard existence result for weak solutions via the methods of characteristics. For this sake we prove existence and uniqueness of the ODE system for $X$ in the Banach space of continuous functions $\mathcal{C}(\Omega)$  with $\Omega = \mathcal{I}^d \times \mathcal{A}$, subsequently using the push-forward of the initial measure. For this sake we first need to verify that the map from $X$ to  the right-hand side in  \eqref{eq:odeX} is indeed Lipschitz-continuous in $\mathcal{C}(\Omega)$ and continuous in time. Since there is no explicit time-dependence and the second component is trivial, this reduces to the Lipschitz-continuity of the first component, which can be verified by straight-forward estimates: 
\begin{lem}\label{lem:uX_lipcont}
    Let $f_0 \in \mathcal{M}(\mathcal{Q})$, and let \ref{it:lip} hold. Then, the map
    \begin{align*}
        &\mathcal{U}: C(\Omega) \rightarrow C(\Omega),\\ &X \mapsto \intia  \phi_{X(\cdot)_x X(y_0,\eta)_x X(\cdot)_\alpha}\odot (X(y_0,\eta)_x-X(\cdot)_x) \dd f_0(y_0,\eta)
    \end{align*} is well-defined and Lipschitz continuous.
\end{lem}

\begin{thm}
    Let $f_0 \in  \mathcal{M}(\mathcal{Q})$ and let $\phi$ be Lipschitz continuous. Then, there exists a  unique solution $f \in L^\infty(0,T;\mathcal{M}(Q))$to \eqref{eq:pdeweak} with initial condition $f_0$.
\end{thm}
\begin{proof}
    Since $\phi$ is Lipschitz continuous, we can use Lemma \ref{lem:uX_lipcont} and, thus, apply Picard-Lindel\"{o}f's Theorem \cite[Theorem 8.13]{kelley_theory_2010} on $\mathcal{C}(\Omega)$ to \eqref{eq:odeX} from which we get existence and uniqueness of a solution continuos in time $X((x_0,\alpha),\cdot)$ to \eqref{eq:odeX} for any $(x_0,\alpha)\in\Omega$. By defining $f$ as the push forward measure along the characteristics, we know that is preserves non-negativity and conserves mass. And by existence and uniqueness of the characteristics it follows that $f$ exists and is unique.
\end{proof}

Let us mention that under our assumptions, the measure $f$ is indeed a solution of the continuity equation with vector field having bounded divergence. Thus, we may even conclude with standard results \cite{diperna1989ordinary} that $f$ is renormalized solutions and $p$-integrability of the solution is preserved in time. For simplicity and brevity we do not further dive into this topic.

We continue by showing that that \eqref{eq:pdeweak} is non-negativity preserving and mass conserving, which are two properties important for probability measures.
\paragraph{Conservation of mass} Let $f$ be a solution to \eqref{eq:pdeweak}, then the total mass is preserved, i.e. $\frac{\dd}{\dd t}\intia f(x,\alpha,t) \dd (x,\alpha)=0$. This follows by using the weak formulation of our PDE \eqref{eq:pdeweak} with test function $\xi\equiv 1$.
\paragraph{Non-negativity of solutions}
The non-negativity of the solution follows directly from the existence and uniqueness of the characteristics and the definition of the solution as push-forward measure of the initial condition along the characteristics.


\paragraph{}From the existence of unique solutions along characteristics, the mass conservation and the non-negativity we obtain the following theorem:

\begin{thm}
    Let \ref{it:lip} hold. For any initial condition $f_0\in\mathcal{P}(\mathcal{Q})$, there exists a unique solution $f\in\mathcal{C}([0,T];\mathcal{P}(\mathcal{Q}))$ to \eqref{eq:pdeweak}.
\end{thm}

\subsection{Evolution of the moments}
In this subsection we investigate the evolution of the mean and variance of solutions. We consider the evolution for either different or equal importance weights first, and illustrate our results with examples at the end.

\subsubsection{Evolution of the mean} We recall the definition of the mean opinion
\begin{align}\label{eq:mu}
\mu(t) := \intia xf(x,\alpha,t) \dd (x,\alpha).
\end{align}
By using the weak formulation of the PDE \eqref{eq:pdeweak} with $\xi(x,\alpha)=x_a$ for any $a \in \{1,\dots,d\}$, we obtain
\begin{align*}
    \frac{\dd}{\dd t}\intia x_a f(x,\alpha,t) \dd(x,\alpha)
=&\intiaia (\nabla_x x_a)
\cdot (\phi_{xy\alpha}\odot (y-x)) f(x,\alpha,t)f(y,\eta,t) \dd (y,\eta) \dd (x,\alpha)\\
=&\intiaia \phi_{{xy\alpha}_a} (y_a-x_a) f(x,\alpha,t)f(y,\eta,t) \dd (y,\eta) \dd (x,\alpha)\\
=&\intiaia  \phi_{{xy\alpha}_a} y_a f(y,\eta,t) \dd (y,\eta)f(x,\alpha,t) \dd (x,\alpha)\\ &-\intiaia  \phi_{{xy\alpha}_a} x_a f(x,\alpha,t) \dd (x,\alpha)  f(y,\eta,t) \dd (y,\eta).
\end{align*} 
We see that the mean and the total mean are in general not preserved. In the special case of everyone having the same importance weights, the mean is preserved since the above calculation reduces to
\begin{alignb}\label{eq:mean_fix_alpha}
    \frac{\dd}{\dd t}\inti x_a \rho(x,t) \dd x
=&\intii  \phi_{{xy}_a} y_a \rho(y,t) \dd y \rho(x,t) \dd x -\intii  \phi_{{xy}_a} x_a \rho(x,t) \dd x  \rho(y,t) \dd y \\
=&\intii \phi_{{xy}_a} y_a \rho(y,t) \dd y \rho(x,t) \dd x- \intii  \phi_{{yx}_a} x_a \rho(x,t) \dd x \rho(y,t) \dd y \\
=&0.
\end{alignb}
\subsubsection{Evolution of the variance}\label{sec:standdev} 
Next recall the definition of the variance
\begin{alignb}\label{eq:standdev}
v(t):=\intia |x-\mu(t)|^2f(x,\alpha,t) \dd (x,\alpha).
\end{alignb}
The time derivative of the variance can be written as 
\begin{align*}
\frac{\dd v(t)}{\dd t}
=&\intiaia \begin{pmatrix}x - \mu(t),y- \mu(t)\end{pmatrix}\Phi_{x\alpha z\eta}\begin{pmatrix}x - \mu(t)\\ y - \mu(t)\end{pmatrix} f(y,\eta,t)f(x,\alpha,t) \dd (x,\alpha) \dd (y,\eta),
\end{align*}
with
\begin{align*}
    \Phi_{x\alpha z\eta}:=\begin{pmatrix}
    -\phi_{xy\alpha} & \frac{\phi_{xy\alpha}+\phi_{yx\eta}}{2} \\
    \frac{\phi_{xy\alpha}+\phi_{yx\eta}}{2} & -\phi_{yx\eta} \\
    \end{pmatrix}.
\end{align*}
The eigenvalues of $\Phi_{x\alpha z\eta}$ are given by 
\begin{align*}
    \lambda_{\pm}
    =&-\frac{\xya+\yxe}{2}\pm \sqrt{\frac{\xya^2+\yxe^2}{2}}. 
\end{align*}
Since $\xya \geq 0$ and $\yxe\geq 0$, the smaller eigenvalue is non-positive since
\begin{align*}
    \lambda_{-}=-\frac{\xya+\yxe}{2}- \sqrt{\frac{\xya^2+\yxe^2}{2}}\leq 0,
\end{align*}
with equality if and only if $\xya=\yxe=0$.
It is straightforward to see that $\lambda_+$ is non-negative and that it is equal to $0$ if and only if $\xya=\yxe$.
Thus, $\Phi_{x\alpha z\eta}$ is negative semi definite if and only if $\xya=\yxe$. Note that if all people have the same importance weights, this condition is satisfied at all points $x,y\in \mI^d$, while it is violated in general when people have different importance weights.\\

If everyone has the same importance weights $\alpha\in\Omega$, then $\xya=\phi_{yx\alpha}$ for all $x,y\in\mI^d$ and thus, from the computations above, it follows that the variance does not increase over time. We can also compute this via
\begin{alignb}\label{eq:standdev_fixalpha}    
 \frac{\dd}{\dd t}v(t)=&\frac{\dd}{\dd t}\inti |x-\mu|^2\rho(x,t) \dd x\\
 =&\intii 2(x-\mu) \cdot ((\phi_{xy}\odot (y-x)) \rho(x,t)\rho(y,t) \dd y \dd x\\
        =&2\intii (x - \mu)\cdot (\phi_{xy}\odot (y- \mu)) \rho(x,t)\rho(y,t) \dd y \dd x\\
        &-2\intii (x - \mu)\cdot (\phi_{xy}\odot (x- \mu)) \rho(x,t)\rho(y,t) \dd y \dd x\\
        =& -\intii \sum_{a=1}^d((x_a-\mu_a) - (y_a-\mu_a))^2\phi_{xy_a} \rho(x,t)\rho(y,t) \dd y \dd x\\
        =& -\intii \sum_{a=1}^d(x_a - y_a)^2\phi_{xy_a} \rho(x,t)\rho(y,t) \dd y \dd x\\
        \leq& 0.
\end{alignb}
Note that this calculation will turn out to be a useful expression in the investigation of the stationary solutions later.
In the special case $\phi \equiv 1$, the following calculation shows that $v$ decreases exponentially 
\begin{alignb}\label{eq:std_phi=1}
    \frac{\dd}{\dd t} v(t) 
    =&\frac{\dd}{\dd t}\intia |x-\mu|^2 f(x,\alpha,t) \dd (x,\alpha)\\
    =&\intiaia 2(x-\mu) \cdot (y-x) f(x,\alpha,t)f(y,\eta,t) \dd (y,\eta) \dd (x,\alpha)\\
=& 2\intiaia (x - \mu)\cdot (y- \mu) f(x,\alpha,t)f(y,\eta,t) \dd (y,\eta) \dd (x,\alpha)\\
&- 2\intiaia (x - \mu)\cdot (x- \mu) f(x,\alpha,t)f(y,\eta,t) \dd (y,\eta) \dd (x,\alpha)\\
=& 2\intia (x - \mu) f(x,\alpha,t) \dd (x,\alpha)\cdot \intia (y- \mu) f(y,\eta,t) \dd (y,\eta)\\
&- 2\intia (x - \mu)\cdot (x- \mu) f(x,\alpha,t) \intia f(y,\eta,t) \dd (y,\eta) \dd (x,\alpha)\\
=& 2(\mu - \mu)\cdot (\mu - \mu)- 2\intia (x - \mu)\cdot (x- \mu) f(x,\alpha,t)\dd (x,\alpha)\\
=& -2v(t).
\end{alignb}

\begin{rem}
    It is also possible to consider a mean and variance which depend on $\alpha$, i.e.
    \begin{align}\label{eq:tildemu}
\Tilde{\mu}(t,\alpha) := \inti xf(x,\alpha,t) \dd x,
    \end{align}
    and
    \begin{align}\label{eq:tildev}
        \Tilde{v}(t,\alpha):=\inti |x-\Tilde{\mu}(t,\alpha)|^2f(x,\alpha,t) \dd x.
    \end{align}
In the case of people having the same opinion weights, these quantities coincide with the original definition of the mean and variance. When people have different opinion weights, $\Tilde{\mu}(t,\alpha)$ can change as we will see in Example \ref{ex:diffmean}.
\end{rem}

We recall that in the case of different importance weights
the variance can increase. We conclude by presenting explicit examples of such cases.
\begin{ex}\label{counterexample_standartdev}
    Let $d=2$ and consider a distribution of the form 
    \begin{align}  
    \label{eq:ex1}
      f(x,\alpha,t=T)=&\frac{1}{30}\delta_{((-\frac{5}{6},1),(\frac{4}{5},\frac{1}{5}))}(x,\alpha)+\frac{1}{30}\delta_{((-1,-1),(\frac{1}{2},\frac{1}{2}))}(x,\alpha)+\frac{14}{15}\delta_{((1,-1),(\frac{1}{2},\frac{1}{2}))}(x,\alpha).  
\end{align}
Set $\beta=\frac{1}{2}$ and choose a smoothed bounded confidence function $\phi(r)$ with $r_1=\frac{2}{5}$ and $r_2=\frac{1}{2}$ in \eqref{eq:phismooth}.
Then, \begin{align*}
\phi_{(-\frac{5}{6},1),(-1,-1),(\frac{4}{5},\frac{1}{5})}=(1,0)
\end{align*}
and \begin{align*}\phi_{(\frac{5}{6},1),(1,-1),(\frac{4}{5},\frac{1}{5})}=\phi_{(1,-1),(\frac{5}{6},1),(\frac{1}{2},\frac{1}{2})}=\phi_{(-1,-1),(1,-1),(\frac{1}{2},\frac{1}{2})}=\phi_{(1,-1),(-1,-1),(\frac{1}{2},\frac{1}{2})}=(0,0).\end{align*}
Thus,
\begin{align*}  
    \left.\frac{\dd}{\dd t}\intia x_1 f(x,\alpha,t) \dd(x,\alpha)\right|_{t=T}
    =&\intiaia  \phi_{{xy\alpha}_1} (y_1-x_1) f(y,\eta,T) \dd (y,\eta)f(x,\alpha,T) \dd (x,\alpha)\\
    =&\frac{1}{900}\left(\phi_{(-\frac{5}{6},1),(-1,-1),(\frac{4}{5},\frac{1}{5})_1} \left(-1+\frac{5}{6}\right)\right)\\
    =&-\frac{1}{5400}\neq0,
\end{align*}
which shows that when people have different importance weights $\alpha$, the mean changes in time. Note that in the above example only the mean in $x_1$ is changing, while the mean in $x_2$ is constant and thus, also the total mean changes. \\
Moreover, the variance is increasing since
\begin{align*}
       -\frac{5}{6}-\mu_1(T)=& -\frac{5}{6}-\intia x_1f(x,\alpha,T) d(x,\alpha)=-\frac{5}{6} -\frac{157}{180}<0,
\end{align*}
and thus $\left.\frac{\dd}{\dd t}  v(t)\right|_{t=T}
    = \frac{-1}{2700}(-\frac{5}{6}-\mu_1(T))>0.$
A  simulation of the respective dynamics is shown in Figure \ref{fig:varalp}.
\end{ex}

\begin{ex}\label{ex:diffmean}
    If we consider the $\alpha$ dependent mean and variance, given by \eqref{eq:tildemu} and \eqref{eq:tildev}, we observe a similar increase as in Example \ref{counterexample_standartdev}. In particular
\begin{align*}  
    \left.\frac{\dd}{\dd t} \Tilde{\mu}_1\left(t, \left(\frac{4}{5},\frac{1}{5}\right)\right)\right|_{t=T}
    =&\left.\frac{\dd}{\dd t}\inti x_1 f\left(x,\left(\frac{4}{5},\frac{1}{5}\right),t\right) \dd x\right|_{t=T}\\
    =&\frac{1}{30}\intia  \phi_{{(-\frac{5}{6},1)y(\frac{4}{5},\frac{1}{5})}_1} \left(y_1+\frac{5}{6}\right) f(y,\eta,T) \dd (y,\eta)\\
    =&\frac{1}{900}\left(\phi_{(-\frac{5}{6},1),(-1,-1),(\frac{4}{5},\frac{1}{5})_1} \left(-1+\frac{5}{6}\right)\right)\\
    =&-\frac{1}{5400}\neq0.
\end{align*}
Also the on $\alpha$ dependent variance can increase when people have different importance weights which we see when changing the opinion weights in the last term of $f$ in Example \ref{counterexample_standartdev}, and thus consider 
    \begin{align*}  
      f(x,\alpha,t=T)=&\frac{1}{30}\delta_{((-\frac{5}{6},1),(\frac{4}{5},\frac{1}{5}))}(x,\alpha)+\frac{1}{30}\delta_{((-1,-1),(\frac{1}{2},\frac{1}{2}))}(x,\alpha)+\frac{14}{15}\delta_{((1,-1),(\frac{4}{5},\frac{1}{5}))}(x,\alpha).  
\end{align*}
We choose all the other parameters like in Example \ref{counterexample_standartdev}. And since then also \begin{align*}\phi_{(1,-1),(-1,-1),(\frac{4}{5},\frac{1}{5})}=(0,0),\end{align*} it follows similarly to above that
\begin{align*}
       -\frac{5}{6}-\Tilde{\mu}_1\left(T,\frac{4}{5},\frac{1}{5}\right)=& -\frac{5}{6}-\inti x_1f\left(x,\left(\frac{4}{5},\frac{1}{5}\right),T\right) \dd x=-\frac{5}{6} -\frac{163}{180}<0,
\end{align*}
and thus $\left.\frac{\dd}{\dd t}  \Tilde{v}(t)\right|_{t=T}
    = \frac{-1}{2700}\left(-\frac{5}{6}-\Tilde{\mu}_1\left(T,\left(\frac{4}{5},\frac{1}{5}\right)\right)\right)>0.$
\end{ex}

\subsection{Maximum component-wise distance in opinion is non-increasing}
The following proposition shows that any solution to \eqref{eq:pdeweak} stays inside any hyper-rectangle that includes the component-wise the maximum and minimum opinion that people have.

\begin{pro}\label{pre:minmax_nondecr}
    Let $f$ be a solution to \eqref{eq:pdeweak} and define 
    \begin{align}\label{eq:j_f}
   \mathcal{J}_{f(t)}:=\{x\in\mI^d\mid\exists \alpha \in \mathcal{A} \text{ s.t. } f(x,\alpha,t)>0\}. 
\end{align} 
Then, for any d-dimensional hyper-rectangle $\mathcal{H}$ with $\mathcal{J}_{f(0)}\subseteq\mathcal{H}$, it holds that  $\mathcal{J}_{f(t)}\subseteq\mathcal{H}$ for all $t\in\mR_{\geq 0}$. 
\end{pro}
\begin{proof}
Let $\mathcal{H}$ be a hyper-rectangle satisfying  $\mathcal{J}_{f(0)}\subseteq\mathcal{H}$.
We want to prove the claim by showing that the characteristics at the boundary of $\mathcal{H}$ point inside. For this let us assume that up to time $t\in\mR_{\geq 0}$, $\mathcal{J}_{f(t)}\subseteq\mathcal{H}$. Let $\alpha\in\mathcal{A}$ and let $n$ denote the outer unit normal vector (as defined in Section \ref{sec:boundary_behabvior}). Note that for any $x$ on the boundary of $\mathcal{H}$ we have that $x_a \leq y_a$ for all $y\in\mathcal{J}_{f(t)}$ and $n_a<0$, or either $x_a \geq y_a$ for all $y\in\mathcal{J}_{f(t)}$ and $n_a\geq0$ or $n_a=0$ for any \ain. Thus, similarly to Section \ref{sec:boundary_behabvior},
 \begin{align*}
        \left(\intia  \phi_{xy\alpha}\odot (y-x) f(y,\eta,t)\dd (y,\eta)\right)\cdot n\leq 0.
\end{align*}

\end{proof}
This implies the following corollary about the maximum component-wise distance in opinion.


\begin{cor}\label{cor:minmax_nondecr}
    Let $f$ be a solution to \eqref{eq:pdeweak}. Then, the maximum component-wise distance in opinion, \begin{align*} \sup_{x,y\in\mathcal{J}_{f(t)}}\sup_{a\in\{1,\dots,d\}}|x_a-y_a|, 
\end{align*}
is non-increasing.
\end{cor}
\begin{proof}
    Let $f$ be a solution to \eqref{eq:pdeweak} and let $T\in\mR_{>0}$ be fixed. Notice that we can built a hyper-rectangle $\mathcal{H}$ by taking as edges in each dimension $a\in\DD$ the topic wise supremum $s_a$ and infimum $i_a$ of opinion, i.e. $$s_a=\sup \left(\{x_a\in\mI| \exists \alpha\in\mathcal{A}, \exists y\in\mI^d \text{ with } y_a= x_a \text{ s.t. } f(y,\alpha,T)>0\}\right)$$ and $$i_a=\inf \left(\{x_a\in\mI| \exists \alpha\in\mathcal{A}, \exists y\in\mI^d \text{ with } y_a= x_a \text{ s.t. } f(y,\alpha,T)>0\}\right)$$. Clearly, $\mathcal{J}_{f(T)}\subseteq \mathcal{H}$, and thus, from Proposition \ref{pre:minmax_nondecr}, it follows that for all $t\geq T$ $\mathcal{J}_{f(t)}\subseteq \mathcal{H}$. This implies that for all $t\geq T$ and for all $a\in\DD$ $$\sup_{x\in\mathcal{J}_{f(t)}}x_a \leq \sup_{x\in\mathcal{J}_{f(T)}}x_a$$ and $$\inf_{x\in\mathcal{J}_{f(t)}}x_a \geq \inf_{x\in\mathcal{J}_{f(T)}}x_a,$$ which implies that $$\sup_{x,y\in\mathcal{J}_{f(t)}}|x_a-y_a|\leq \sup_{x,y\in\mathcal{J}_{f(T)}}|x_a-y_a|.$$
Since this holds for all $a\in\DD$, we get that for all $t\geq T$
    $$\sup_{x,y\in\mathcal{J}_{f(t)}}\sup_{a\in\{1,\dots,d\}}|x_a-y_a|\leq \sup_{x,y\in\mathcal{J}_{f(T)}}\sup_{a\in\{1,\dots,d\}}|x_a-y_a|,$$
    which proves the claim.
\end{proof}


\section{Stationary solutions}\label{sec:statsol}

Next we investigate possible stationary states of \eqref{eq:pdeweak}.
We say that an $f_\infty\in\mathcal{P}(\mathcal{Q})$ is a stationary solution of \eqref{eq:pdeweak} if it does not depend on time $t$ and it satisfies \eqref{eq:pdeweak}.
We will see that stationary solutions $f_\infty$ can be of the following forms: 
\begin{enumerate}[label={\textbf{(S\arabic*)}}]
    \item \textit{Consensus}; a single concentrated point measure (Dirac measure) in opinion space. (It does not need to be concentrated in importance space.)
    \item \textit{Separated clusters}; multiple Dirac measures in opinion space that are located so far from each other that no interactions are happening, i.e. $\phi_{xy\alpha}=0$ for all $x,y\in\mI^d, x\neq y$ $\alpha,\eta\in\mathcal{A}$ with 
    $(x,\alpha),(y,\eta)\in\text{supp}(f_\infty)$.
    \item\label{it:interactingcluster} \textit{Interacting clusters}; multiple interacting Dirac measures in opinion space, located in such as way that interactions cancel out. This means that there exist some  $x,y\in\mI^d, x\neq y$ $\alpha,\eta\in\mathcal{A}$ with
    \begin{align*}
        (x,\alpha),(y,\eta)\in\text{supp}(f_\infty) \text{ and } \phi_{xy\alpha}>0 \text{ as well as } \frac{\dd f_\infty(x,\alpha)}{\dd t}=0
    \end{align*}
       for all $(x,\alpha)\in\mathcal{Q}$. We will give an example of such an interacting cluster in Example \ref{ex:statsol}.

\end{enumerate}

The three stationary states are constructed explicitly and verified as stationary solutions. We do not claim this list is exhaustive, as a full classification of stationary states is outside the scope of this paper.\\
Note that the interacting cluster state corresponds to a stationary configuration not present in classical bounded-confidence models, which only exhibit consensus or fully decoupled cluster formation as long-time behaviour.

\subsection{Consensus formation} 

We start by presenting results which lead to consensus under appropriate assumptions. First we consider the simplest case, i.e. $\phi_{x,y,\alpha}\equiv 1$. Clearly, if
\begin{align*}
    0=\Bigl(\inti (y-x) f_\infty(y)\dd y\Bigr) f_\infty(x) =  (\mu-x) f_\infty(x)
\end{align*}
holds, then $f_\infty$ is a stationary solution. We recall that $\mu$ corresponds to the mean defined in \eqref{eq:mu}. 
Thus, a stationary solution is given by \begin{align*}
    f_\infty(x)=\delta_{\mu}(x).
\end{align*}
From \eqref{eq:std_phi=1} it follows that this stationary solution is unique. \\

In the following theorem, we show that people reach consensus when everyone interacts on all topics initially. This is a similar result, but different proof, to what has been shown in the discrete case in \cite{cahill_modified_2025}. 

\begin{thm}\label{thm:consensus}
    Let $\phi:[0,2]\to[0,1]$ be monotonically decreasing and $f_0\in\mathcal{P}(\mathcal{Q})$ such that $\phi(p_a(x,y,\alpha))\geq c$ for some $c\in \mR_{>0}$ for all  $a\in \{1,\dots,d\}$ and for all $ x,y\in\mI^d, \alpha \in \mathcal{A}$ with $(x,\alpha)\in\text{supp}(f_0)$ for which there exists an $\eta\in\mathcal{A}$ such that $(y,\eta)\in\text{supp}(f_0)$. Then, any solution $f$ of \eqref{eq:pdeweak} with initial condition $f_0$ converges in any Wasserstein metric $W_p$ ($p \leq 1 \leq \infty$) to $f_\infty(x,\alpha)=\inti f_0(y,\alpha) \dd y\delta_{\mu}(x)$ for some $\mu \in \mI^d$.
\end{thm}
\begin{proof}
Notice that since $\phi$ is monotone decreasing, by Proposition \ref{pre:minmax_nondecr} it follows from the condition on $\phi$ and $f_0$ that for a solution $f$ to \eqref{eq:pdeweak} with initial condition $f_0$, that at any time step \tinb $$\phi(p_a(x,y,\alpha))\geq c$$ for all $  a\in \{1,\dots,d\}$ and for all $x,y\in\mI^d, \alpha \in \mathcal{A}$ with $(x,\alpha)\in\text{supp}(f(.,.,t))$ for which there exists an $\eta\in\mathcal{A}$ with $(y,\eta)\in\text{supp}(f(.,.,t))$.\\
Next we prove convergence in each dimension. Choose \ain arbitrarily and let $$(x^{\min}(t),x^{\max}(t))=\text{argsup}_{x,y\in\mathcal{J}_f(t)}|x_a-y_a|.$$ 
From Proposition \ref{pre:minmax_nondecr}, we know that  $x_a^{\min}$ is non-decreasing in $t$ and $x_a^{\max}$ is non-increasing. Since $x_a^{\min}(t)$ is bounded from above by $x_a^{\max}(t)$ and $x_a^{\max}(t)$ is bounded from below by $x_a^{\min}(t)$, it follows that $x_a^{\min}(t)$ and $x_a^{\max}(t)$ converge, i.e. there exist some $u,v\in\mI$ such that \begin{align*}
    x_a^{\min}(t) \rightarrow u \leq v \leftarrow x_a^{\max}(t).
\end{align*}
We want to show that $u=v$. For this let us assume that $u<v$. If we split the interval $[y,z]$ in half, there needs to be at least half of the mass on one of the two sides, i.e. either
\begin{enumerate}[label=(\roman*)]
    \item\label{i:case1} $\int\limits_\mathcal{A}\int\limits_{-1}^{\frac{u+v}{2}}\int\limits_{\mI^{d-1}}f(x,\alpha,t)\dd (x,\alpha) \geq \frac{1}{2}$ or 
    \item\label{i:case2}  $\int\limits_\mathcal{A}\int\limits_{\frac{u+v}{2}}^{1}\int\limits_{\mI^{d-1}}f(x,\alpha,t)\dd (x,\alpha) \geq \frac{1}{2}$
\end{enumerate} 
In case \ref{i:case1} along a characteristic curve with $x_a(t) \geq v $ it holds that for any \tinb 
\begin{align*}
    -\frac{\dd}{\dd t} x_a(t) =& -\intia  \phi_{x(t) y\alpha_a} (y_a-x_a(t))) f(y,\eta,t)\dd (y,\eta)\\
    \geq& -\int\limits_\mathcal{A}\int\limits_{-1}^{\frac{u+v}{2}}\int\limits_{\mI^{d-1}}\underbrace{\phi_{x(t) y\alpha_a}}_{\geq c}\underbrace{(y_a-x_a(t))}_{\leq -\frac{v-u}{2}}f(y,\eta,t)\dd (y,\eta)\\
    \geq& \frac{c(v-u)}{4}.
\end{align*}
Similarly in case \ref{i:case2}, along a characteristic curve with $x_a(t) \leq u$, we find for any \tinb 
\begin{align*}
    \frac{\dd}{\dd t} x_a(t) =& \intia  \phi_{x(t) y\alpha_a} (y_a-x_a(t)) f(y,\eta,t)\dd (y,\eta)\\
    \geq& \int\limits_\mathcal{A}\int\limits_{-1}^{\frac{u+v}{2}}\int\limits_{\mI^{d-1}}\underbrace{\phi_{x(t) y\alpha_a}}_{\geq c}\underbrace{(y_a-x_a(t))}_{\geq \frac{v-u}{2}}f(y,\eta,t)\dd (y,\eta)\\
    \geq& \frac{c(v-u)}{4}.
\end{align*}
Thus, in both cases, for any \tin, we obtain a linear change in time as long the characteristics are outside the interval $[u,v]$ and hence
\begin{align*}
\left(x_a^{\min}(t)-x_a^{\max}(t)\right)\geq \left(x_a^{\min}(0)-x_a^{\max}(0)\right) + \frac{c(v-u)t}{4}.
\end{align*}
If $ u\neq v$, This is a contradiction since $x_a^{\min}(t)$ and $x_a^{\max}(t)$ are converging.

Thus $u=v$ and $x_a^{\min}(t)$ and $x_a^{\max}(t)$ converge to the same value. Moreover, since $\alpha$ is arbitrary, this holds in every dimension and $f$ converges to one Dirac measure in space. Since $f$ does not change in the importance weight space, $\inti f(y,\alpha,t) \dd y$ does not change in time and thus, $f$ converges to $\inti f_0(y,\alpha) \dd y\delta_{\mu}(x)$ for some $\mu \in \mI^d$. 
\end{proof}

In the case when all people have the same opinion weights, Theorem \ref{thm:consensus} and \eqref{eq:mean_fix_alpha} imply
\begin{cor}
    Let all people have the same opinion weights $\alpha\in\mathcal{A}$. Let $\phi:[0,2]\to[0,1]$ be monotonically decreasing and $\rho_0\in\mathcal{P}(\mI^d)$ such that $\phi(p_a(x,y,\alpha))\geq c$ for some $c\in \mR_{>0}$ for all  $a\in \{1,\dots,d\}$ and for all $x,y\in\text{supp}(\rho_0)$. Then, any solution $\rho$ of \eqref{eq:pdeweak_fixalp} with initial condition $\rho_0$ converges to $\rho_\infty(x)=\delta_{\mu}(x)$, where $\mu$ denotes the mean opinion defined in \eqref{eq:mu}.
\end{cor}

\subsection{Separated clusters}\label{sec:clusterform}

Next we want to investigate stationary solutions of \eqref{eq:pdeweak}, for which clusters do not interact.
It holds that in general any \begin{align*}
    f_\infty(x,\alpha)=\sum_{\ell=1}^M c_\ell \delta_{(z_\ell,\alpha_\ell)}(x,\alpha)\end{align*} for $M \in \mN, z_\ell \in \mI^d, \alpha_\ell \in \mathcal{A}$ for all $l \in \{1,...,M\}$ and for all $l \in \{1,...,M\}$ with $ c_\ell>0$ and  $\sum_{\ell=1}^M c_\ell=1$ is a stationary solution if the $z_\ell$ are spread out sufficiently, i.e. $\phi_{z_\ell z_k \alpha_\ell}\odot (z_k-z_\ell)=0$ for all $l,k \in \{1,...,M\}$.
We can prove this by plugging this $f_\infty$ in the weak formulation \eqref{eq:pdeweak} 
\begin{align*}
    \frac{\dd}{\dd t}\intia \xi(x,\alpha) &f_\infty(x,\alpha) \dd (x,\alpha)\\
    =&\intiaia \nabla_x \xi(x,\alpha) \cdot (\phi_{xy\alpha}\odot (y-x)) f_\infty(x,\alpha)f(y,\eta,t) \dd (y,\eta) \dd (x,\alpha)\\
    =&\intiaia \nabla_x \xi(x,\alpha) \cdot (\phi_{xy\alpha}\odot (y-x)) \sum_{\ell=1}^\mathcal{L} c_\ell \delta_{(z_\ell,\alpha_\ell)}(x)\sum_{l=k}^\mathcal{L} c_k \delta_{(z_k,\alpha_k)}(y) \dd (y,\eta) \dd (x,\alpha)\\
    =& \sum_{\ell=1}^\mathcal{\mathcal{L}} c_\ell\sum_{l=k}^\mathcal{L} c_k \nabla_x \xi((z_\ell,\alpha_\ell)) \cdot (\phi_{z_\ell z_k\alpha_\ell}\odot (z_k-z_\ell))\\
    =&0.
\end{align*}
Furthermore, in the case of same importance weights $\alpha$, stationary solution have to be of that form.
This follows from  \eqref{eq:standdev_fixalpha}, in particular
\begin{align*}
    \frac{\dd v}{\dd t} = -\intii \sum_{a=1}^d(x_a - y_a)^2\phi_{xy_a} \rho(x,t)\rho(y,t) \dd y \dd x.
\end{align*}
Thus, if there exist $x,y\in\mI^d$ with $ x\neq y$ such that $\rho(x)>0, \rho(y)>0$ and $\phi_{x,y}\neq 0$, then $\frac{\dd v}{\dd t}<0$. Consequently, $\rho$ can not be a stationary solution and we obtain the following corollary.
\begin{cor}
Let all people have the same importance weights and set
\begin{align}
    \label{eq:f_delta}
    \rho_\infty(x)=\sum_{\ell=1}^M c_\ell \delta_{z_\ell}(x),
\end{align}
with $M \in \mN, z_\ell \in \mI^d$ and $c_\ell>0$ for all $\ell \in \{1,...,M\}$ and $\sum_{\ell=1}^M c_\ell=1$. \\
Then $\rho_\infty(x)$ given by \eqref{eq:f_delta} is a stationary solution to \eqref{eq:pdeweak_fixalp} if and only if $$\phi_{z_\ell z_k}\odot(z_k-z_\ell)=0 \text{ for all } \ell,k \in \{1,...,M\}.$$
\end{cor}

Since it gives a necessary condition, it implies that when people have the same importance weights, the Dirac masses that the stationary solutions consist of have to be located a certain distance apart from each other. And, since the opinion space we consider is bounded, we can compute a bound on the number of Dirac measures. 

\subsubsection{Maximal number of clusters in the case of same importance weights in $2D$}

We want to determine the maximal number of clusters in a stationary solution to \eqref{eq:pdeweak} for $d=2$. Since the interaction radii depend on the $p$-norm and thus on the choice of $\alpha$, we consider the simpler case of equal importance weights, i.e. $\alpha = (\alpha_1, \alpha_2)$ for everyone. Furthermore, we assume that $\beta\geq \frac{1}{2}$ and let $\text{supp}(\phi)\subset [0,R]$ for some $R\in (0,2]$.\\

Consider the p-distance defined in \eqref{eq:p_distance}. Since $\beta \geq \frac{1}{2}$, we have that $\frac{R}{\beta}\geq \frac{R}{(1-\beta)\alpha_1}$ and $\frac{R}{\beta}\geq \frac{R}{(1-\beta)\alpha_2}$. Since the interaction function $\phi$ is compactly supported on $[0,R]$, we can sketch the interaction domain of a person with opinion $(x_1, x_2)$ in Figure \ref{fig:shapealphanorm}. We see that a person having opinion $(x_1, x_2)$ would interact on topic one with all people having opinion vectors in $$\{y\in\mI^2| p_1(x,y,\alpha)< R\},$$ corresponding to the dark purple diamond, and regarding topic two with all people in $$\{y\in\mI^2| p_2(x,y,\alpha)< R\},$$ corresponding to the light purple diamond. In particular, they would interact on both topics with people having opinions in the intersection of those two sets. We can bound that region from below by the set $$\{y\in\mI^2|(\beta + (1-\beta)\alpha_1)|x_1-y_1|+(\beta + (1-\beta)\alpha_2)|x_2-y_2|< R\},$$ displayed as orange diamond, and from above by  $$\{y\in\mI^2|(1-\beta)\alpha_1|x_1-y_1|+(1-\beta)\alpha_2|x_2-y_2|< R\},$$ i.e. the blue diamond, and by $$\{y\in\mI^2|(\beta + (1-\beta)\alpha_1)|x_1-y_1| < R \land (\beta + (1-\beta)\alpha_2)|x_2-y_2|< R\},$$ i.e. the green rectangle.\\

\begin{figure}[h!t]
\centering
    \ctikzfig{bilder/shape_alphanorm_new}
    \caption{Interaction radius defined by the $p_\alpha$-distance \eqref{eq:p_distance} for an interaction function $\phi$ with compact support on $[0,R]$}
    \label{fig:shapealphanorm}
\end{figure}

\noindent The upper and lower bounds on the square $\mI^d=[-1,1]^2$ follow from the following considerations.
\begin{itemize}
    \item Upper bound: the maximum number of Dirac measures is bounded from above by the maximum number of orange diamonds fitting into $[-1,1]^2$, i.e. 
     $$2\left(\lfloor\frac{2(\beta+(1-\beta)\alpha_1)}{R}\rfloor+1\right)\left(\lfloor\frac{2(\beta+(1-\beta)\alpha_2)}{R}\rfloor+1\right).$$
     \item Lower bound: Clearly, $1$ is a lower bound. However, we want to improve this bound. Covering $[-1,1]^2$ with blue diamond provides a lower bound, since the blue diamond covers both purple diamonds, and thus the entire region where interactions are happening. This lower bound on the number of clusters is given by $$\left(\lfloor\frac{2(1-\beta)\alpha_1}{R}\rfloor+1\right) \left(\lfloor\frac{2(1-\beta)\alpha_2}{R}\rfloor+1\right) + \lfloor\frac{2(1-\beta)\alpha_1}{R}\rfloor \lfloor\frac{2(1-\beta)\alpha_2}{R}\rfloor.$$
     
     Notice that it is possible to place a person for example at $(x_1+\frac{R}{\beta+(1-\beta)\alpha_1}+\epsilon,x_2$) for some small $\epsilon>0$ such that the $p_2$ norm is so small that the people would interact regarding topic two but clearly not regarding topic one. Since they already agree on topic two though, their opinions do not change. The same behaviour can also be observed in the direction of the first topic. Thus, the number of green rectangles needed to cover $[-1,1]^2$ that are placed with distance $\frac{R}{\beta+(1-\beta)\alpha_1}$ in the direction of topic 1 and  distance $\frac{R}{\beta+(1-\beta)\alpha_2}$ in the direction of topic 2 from each other, provides a lower bound on the maximal number of clusters. This bound is given by $$\left(\lfloor\frac{2(\beta+(1-\beta)\alpha_1)}{R}\rfloor+1\right)\left(\lfloor\frac{2(\beta+(1-\beta)\alpha_2)}{R}\rfloor+1\right).$$
\end{itemize}
Coming back to the computations done at the beginning in Section \ref{sec:clusterform}, we see that even when people have different importance weights, the following more general but also weaker corollary holds.
\begin{cor}\label{cor:var_alp_stat_sol}
Let
\begin{align}
\label{eq:f_delta_2}
f_\infty(x,\alpha)=\sum_{\ell=1}^M c_\ell \delta_{(z_\ell,\alpha_\ell)}(x,\alpha)\end{align}
with $M \in \mN, z_\ell \in \mI^d, \alpha_\ell \in\mathcal{A}$ and $
c_\ell>0$ for all $\ell \in \{1,...,M\}$ and $\sum_{\ell=1}^M c_\ell=1$. \\
Then $f_\infty(x,\alpha)$ given by \eqref{eq:f_delta_2} is a stationary solution to \eqref{eq:pdeweak} if $$\phi_{z_\ell z_k \alpha_\ell}\odot(z_k-z_\ell)=0 \text{ for all } \ell,k \in \{1,...,M\}.$$ 
\end{cor}
Note that Corollary \ref{cor:var_alp_stat_sol} is a sufficient but not necessary condition. This motivates the next part where we look into stationary states that have a different form.

\subsection{Interacting clusters}

We conclude with two examples illustrating the existence of interacting clusters \ref{it:interactingcluster} when people can have different importance weights. Furthermore, we provide an example showing that the distance between the location of the Dirac measure masses in these interacting clusters can be arbitrarily close. 
\begin{ex}\label{ex:statsol} Consider
    \begin{align}\label{eq:fake_ss}
    f_\infty(x,\alpha)=&\frac{1}{3}\delta_{((-1,-\frac{1}{2}),(1,0))}(x,\alpha)+\frac{1}{3}\delta_{((0,0),(0,1))}(x,\alpha)+\frac{1}{3}\delta_{((1,\frac{1}{2}),(1,0))}(x,\alpha).
\end{align}
Set $\beta=\frac{1}{2}$ and consider the smoothed interaction function \eqref{eq:phismooth} with $r_1=\frac{1}{2}$ and $r_2=\frac{5}{8}$. We will show that $f_{\infty}$ satisfies the assumption of an interacting cluster.\\
Note that
\begin{align*}
\phi_{(0,0),(-1,-\frac{1}{2}),(0,1)}=\phi_{(0,0),(1,\frac{1}{2}),(0,1)}=&(0,1)\\
\phi_{(-1,-\frac{1}{2}),(0,0),(1,0)}=\phi_{(1,\frac{1}{2}),(0,0),(1,0)}=\phi_{(-1,-\frac{1}{2}),(1,\frac{1}{2}),(1,0)}=\phi_{(1,\frac{1}{2}),(-1,-\frac{1}{2}),(1,0)}=&(0,0)
\end{align*}
and define
\begin{align}
\label{eq:S}
    S(x,\alpha, f):=\Bigl(\intia  \phi_{xy\alpha}\odot (y-x) f(y,\eta,t)\dd (y,\eta)\Bigr) f(x,\alpha,t).
\end{align}
A stationary solution $f_{\infty}$ has to satisfy $S(x,\alpha, f_\infty)=(0,0)$. Clearly, for all $(x,\alpha)\in\mathcal{Q}\backslash((0,0),(0,1))$, $S(x,\alpha, f_\infty)=(0,0)$ since there either $f_\infty(x,\alpha)=0$ or $\phi_{xy\alpha}=0$ for all $y\in \mI^d$. In addition, we get
\begin{align*}
    S((0,0),(0,1), f_\infty)=\frac{1}{9}\left((0,1)\odot (-1,-\frac{1}{2})+(0,1)\odot (1,\frac{1}{2})\right)=\frac{1}{9}\left(0,-\frac{1}{2}+\frac{1}{2}\right)=(0,0).
\end{align*}
Therefore \eqref{eq:fake_ss}, also shown in Figure \ref{fig:ex_statsol}, is an interacting cluster.
    \begin{figure}[h!t]
    \centering
    \includegraphics[width=0.4\linewidth]{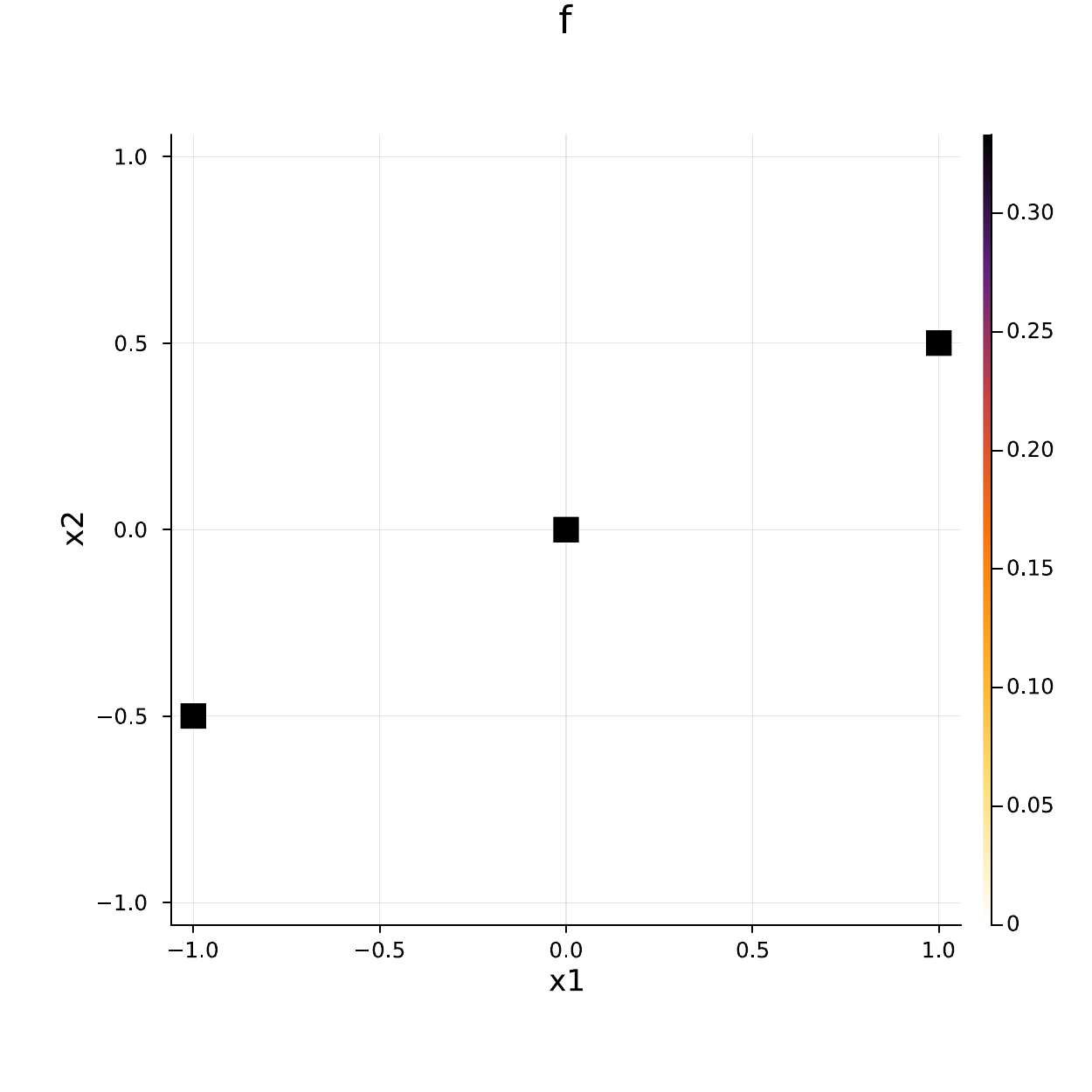}
    \caption{Example of an interacting cluster discussed in Example \ref{ex:statsol}}
    \label{fig:ex_statsol}
\end{figure}
\end{ex}

Note that \eqref{eq:fake_ss} is not a stable stationary solution. To show that, we add a small $\epsilon\in(0,\frac{1}{8})$ to, for example, the Dirac measure at $(-1,-\frac{1}{2})$. Since $\epsilon$ is small, it still holds that
\begin{align}
    \phi_{(0,0),(-1,-\frac{1}{2}),(0,1)}=\phi_{(0,0),(1,\frac{1}{2}),(0,1)}=&(0,1).
\end{align}
However,
\begin{align*}
    S((0,0),(0,1), f_\infty)=\frac{1}{9}\left((0,1)\odot (-1,-\frac{1}{2}+\epsilon)+(0,1)\odot (1,\frac{1}{2})\right)=\frac{1}{9}\left(0,\epsilon\right)\neq (0,0).
\end{align*}

\FloatBarrier
\begin{ex}\label{ex:closeness} In this example we will show that the location of the interacting clusters can be arbitrarily close.\\ 
Let $\varepsilon \in \left(0,\frac{1}{4}\right]$ be arbitrary. As in Example \ref{ex:statsol}, we choose $\beta=\frac{1}{2}$ and a smoothed bounded confidence function $\phi(r)$ with $r_1=\frac{1}{2}$ and $r_2=\min\left(\frac{5}{8},\frac{1}{2}+\epsilon\right)$ in \eqref{eq:phismooth}. We now want to show that
    \begin{alignb}
    \begin{split}
    \label{eq:statsol_epsilon}
    f_\infty(x,\alpha)=&\frac{2\varepsilon}{1+2\varepsilon}\delta_{((-1,-\frac{1}{2}),(1,0))}(x,\alpha)+\frac{1}{4}\delta_{((0,0),(0,1))}(x,\alpha)+\frac{2\varepsilon}{1+2\varepsilon}\delta_{((1,\frac{1}{2}),(1,0))}(x,\alpha)\\ &+\frac{3-10\varepsilon}{8(1+2\varepsilon)}\delta_{((0,\varepsilon),(1,0))}(x,\alpha)+\frac{3-10\varepsilon}{8(1+2\varepsilon)}\delta_{((0,\varepsilon),(1,0))}(x,\alpha).
    \end{split}
\end{alignb}
is a stationary solution. Let us compute
\begin{align*} 
\phi_{(0,0),(-1,-\frac{1}{2}),(0,1)}=\phi_{(0,0),(1,\frac{1}{2}),(0,1)}=&(0,1)\\
\phi_{(-1,-\frac{1}{2}),(0,0),(1,0)}=\phi_{(1,\frac{1}{2}),(0,0),(1,0)}=\phi_{(-1,-\frac{1}{2}),(1,\frac{1}{2}),(1,0)}=\phi_{(1,\frac{1}{2}),(-1,-\frac{1}{2}),(1,0)}=&(0,0)\\
\phi_{(-1,-\frac{1}{2}),(0,\pm \varepsilon),(1,0)}=\phi_{(1,\frac{1}{2}),(0,\pm\varepsilon),(1,0)}=&(0,0)\\
\phi_{(0,-\varepsilon),(-1,-\frac{1}{2}),(0,1)}=\phi_{(0,\varepsilon),(1,\frac{1}{2}),(0,1)}=&(0,1)\\ \phi_{(0,\varepsilon),(-1,-\frac{1}{2}),(0,1)}=\phi_{(0,-\varepsilon),(1,\frac{1}{2}),(0,1)}=&(0,0)\\
\phi_{(0,\varepsilon),(0,-\varepsilon),(0,1)}=\phi_{(0,\pm\varepsilon),(0,0),(0,1)}=&(1,1).
\end{align*}
Clearly, for all $(x,\alpha)\in\mathcal{Q}\backslash\{((0,0),(0,1)),((0,\pm \varepsilon),(0,1))\}$, $S(x,\alpha, f_\infty)=(0,0)$, as defined in \eqref{eq:S}, since there either $f_\infty(x,\alpha)=0$ or the $\phi_{xy\alpha}=0$ for all $y\in \mI^d$. In addition,
\begin{align*}
    S((0,0),(0,1),f_\infty)
    &=\frac{1}{4}\left(\frac{2\varepsilon}{1+2\varepsilon}\left((0,1)\odot (-1,-\frac{1}{2})+(0,1)\odot (1,\frac{1}{2})\right) \right. \\
    & \qquad \left.+\frac{3-10\varepsilon}{8(1+2\varepsilon)}\left((1,1)\odot (0,-\varepsilon)+(1,1)\odot (0,\varepsilon)\right)\right)\\
    &=(0,0),
    \end{align*}
    and
    \begin{align*}
    S((0,\pm\varepsilon),(0,1), f_\infty)&=\frac{3-10\varepsilon}{8(1+2\varepsilon)}\left(\frac{2\varepsilon}{1+2\varepsilon}(0,1)\odot (\pm 1,\pm \frac{1}{2}\mp \varepsilon)+\frac{1}{4}(1,1)\odot (0,\mp \varepsilon) \right. \\
    & \qquad \left. + \frac{3-10\varepsilon}{8(1+2\varepsilon)}(1,1)\odot (0,\mp 2\varepsilon)\right)\\
    &=(0,0).
\end{align*}
Thus, $f_\infty(x,\alpha)$ defined in \eqref{eq:statsol_epsilon} and shown in Figure \ref{fig:ex_statsol_epsi} is a stationary solution, in which the interacting clusters are arbitrarily close.

    \begin{figure}[h!t]
    \centering
    \includegraphics[width=0.5\linewidth]{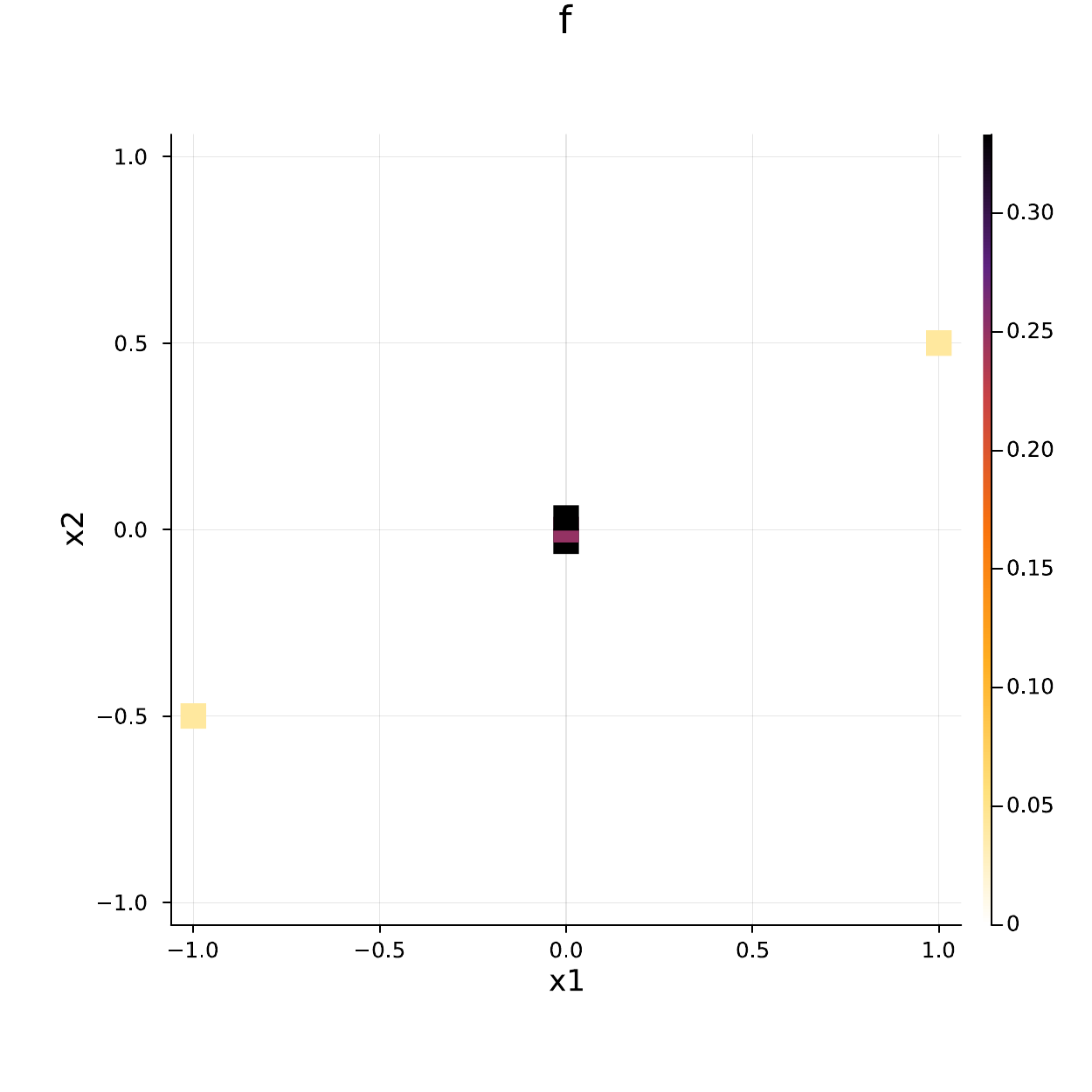}
    \caption{Interacting cluster \eqref{eq:statsol_epsilon}, in which the Dirac measures are arbitrary close (as discussed in Example \ref{ex:closeness}).}
    \label{fig:ex_statsol_epsi}
\end{figure}
\end{ex}

\FloatBarrier

\section{Simulations}\label{sec:numerics}

We now illustrate the dynamics of \eqref{eq:pde} using computational experiments. In doing so we approximate
\( f(x, \alpha, t) \) by the empirical density
\[
f(x, \alpha, t) \approx \frac{1}{N}\sum_{i=1}^N \delta_{x_i(t)}(x)\, \delta_{\alpha_i}(\alpha),
\]
where \( x_i(t) \) is the position of particle \( i \) at time \( t \), and \( \alpha_i \) is its importance weight. The evolution of the particle positions is governed by the ODE system
\begin{align}\label{eq:ode}
\frac{d}{dt} x_i(t) = \frac{1}{N} \sum_{\substack{j=1 \\ j \neq i}}^N \phi_{x_i x_j \alpha_i} \odot (x_j - x_i) \qquad \text{ for } i = 1, \ldots N.
\end{align}
The initial positions \( x_i(0) \) are computed from the initial particle distribution \( f_0(x, \alpha) \). In particular, we discretize the domain into $65$ grid points in each direction of opinion space, and a set of parameter values \( \{\alpha_l\} \).
At each grid point $x_k$, we compute the initial density \( f_0(x_k, \alpha_l) \) and place \( n_{k,l} = \mathrm{round}(f_0(x_k, \alpha_l) \cdot s) \) particles at position \( x_k \) with parameter \( \alpha_l \), where \( s \) is a scaling factor controlling the total number of particles.
All particles have the same weight, i.e. $
w_i = \frac{1}{N}
$ where \( N \) is the total number of particles. We solve \eqref{eq:ode} using the Julia package solver "Vern9()", see \cite{ODESolver}. Vern9 is "Verner's “Most Efficient” 9/8 Runge-Kutta method", which is characterised by its high accuracy and stability.

\subsection{Opinion dynamics for different distance functions}\label{sec:sim_ex1}
In the following we discuss the impact of the distance used to measure 'closeness in opinion' on the dynamics and the stationary states of \eqref{eq:pde}. We demonstrate that for the Euclidean distance, the component-wise distance and the $p_\alpha$-distance \eqref{eq:p_distance} with same $\alpha$ for all people, the observable dynamics are rather simple and interactions are symmetric while when choosing the $p_\alpha$-distance \eqref{eq:p_distance} and assigning different importance weights $\alpha$ to different people, the dynamics are more complex and new behaviours occur. In particular we consider the distances
\begin{enumerate}[label=\textbf{(D\arabic*)}]
    \item\label{it:p_v} $p_\alpha$-distance \eqref{eq:p_distance} with varying importance weights $\alpha_i$
    \item\label{it:p_f} $p_\alpha$-distance \eqref{eq:p_distance} with the same $\alpha_i = \alpha$ for each person
    \item\label{it:comp_eucl} Component-wise distance, i.e. in dimension $a\in\DD$ the distance between $x,y\in\mI^d$ is $|x_a-y_a|$, which corresponds to a Hegselmann-Krause model \cite{hegselmannkrause} in each dimension
    \item\label{it:eucl} Euclidean distance, i.e. the distance between $x,y\in\mI^d$ is $\sqrt{\sum_{a=1}^d|x_a-y_a|^2}$ as in \cite{fortunato_vector_2005}.
\end{enumerate}
We choose the initial distribution as in Example \ref{counterexample_standartdev}, i.e. 
    \begin{align}  
    \label{eq:ex1_new}
      f_0(x,\alpha)=&\frac{1}{30}\delta_{((-\frac{5}{6},1),(\frac{4}{5},\frac{1}{5}))}(x,\alpha)+\frac{1}{30}\delta_{((-1,-1),(\frac{1}{2},\frac{1}{2}))}(x,\alpha)+\frac{14}{15}\delta_{((1,-1),(\frac{1}{2},\frac{1}{2}))}(x,\alpha).  
\end{align}
If everyone has the same $\alpha$ (case \ref{it:p_f}) or if the distance does not depend on $\alpha$ as in case \ref{it:comp_eucl} and case \ref{it:eucl}, we use the initial condition
\begin{align}\label{eq:ex1_noalp}
    \rho_0(x)=&\frac{1}{30}\delta_{(-\frac{5}{6},1)}(x)+\frac{1}{30}\delta_{(-1,-1)}(x)+\frac{14}{15}\delta_{(1,-1)}(x).
\end{align}
Table \ref{tab:parameterchoice_difdist} lists all parameters used for the simulations.
\begin{table}[ht]
    \centering
    \begin{tabular}{c c c}
    \toprule
    \textbf{Parameter} & \textbf{Notation} & \textbf{Value} \\
      \midrule
        number of topics & $d$ & $2$ \\[2pt]
         ratio of current vs other topics & $\beta$ & $\frac{1}{2}$\\[2pt]
         lower bound for \eqref{eq:phismooth} & $r_1$ & $\frac{2}{5}$ \\[2pt]
         upper bound for \eqref{eq:phismooth} & $r_2$ & $\frac{1}{2}$ \\[2pt]
        final time & $T$ & $2500$ \\[2pt]
        scaling factor & $s$ & $25$ \\[2pt]
        number of particles & $N$ & $25$ \\
    \bottomrule    
    \end{tabular}
    \caption{Parameters used in Subsection \ref{sec:sim_ex1}}
    \label{tab:parameterchoice_difdist}
\end{table}
The outcomes of the simulations for the different distances are shown in Figure \ref{fig:difdis}. We can see in Figure \ref{fig:varalp} that, when using the $p_\alpha$-distance and people have different importance weights, \ref{it:p_v}, it is possible for some people to interact with people who do not interact with them, i.e. the interactions do not have to be symmetric. This dynamic is different to all the other cases we investigated. In the case \ref{it:p_f} we see in Figure \ref{fig:alp_0.8} and \ref{fig:alp_0.5} that whether or not the people with opinions $(-\frac{5}{6},1)$ and $(-1,-1)$ interact with each other depends on the value of $\alpha$. In particular, they interact with each other on the first topic if $\alpha=(\frac{4}{5},\frac{1}{5})$ and do not interact if $\alpha=(\frac{1}{2},\frac{1}{2})$. This is in contrast to the case in which people have different importance weights \ref{it:p_v}, in which interactions occur in all opinions or not at all. When using the Hegselmann-Krause model in two dimensions, i.e. case \ref{it:comp_eucl}, we see that similar to the case where $\alpha_1>>\alpha_2$, people with opinions $(-\frac{5}{6},1)$ and $(-1,-1)$ interact on the first topic with interactions being again reciprocal. Furthermore, in the case of \ref{it:comp_eucl} opinions on different topic do not influence the others. Therefore it is not possible to observe dynamics arising from the interplay between different topics. In Figure \ref{fig:eucli}, we used the Euclidean norm as a distance measure, i.e. case \ref{it:eucl}. We see that there are no interactions happening (for that choice of $\phi$). This is caused by the fact that opinions $(-\frac{5}{6},1)$ and $(-1,-1)$ are close in the first component, but not the second one. Note that a much larger interaction radius $r_1$ will lead to interactions. 
Again, in case \ref{it:eucl} people either interact in all opinions or do not interact at all.

\begin{figure}[!ht]  
  \begin{subfigure}[t]{\textwidth}
    \includegraphics[width=0.495\textwidth]{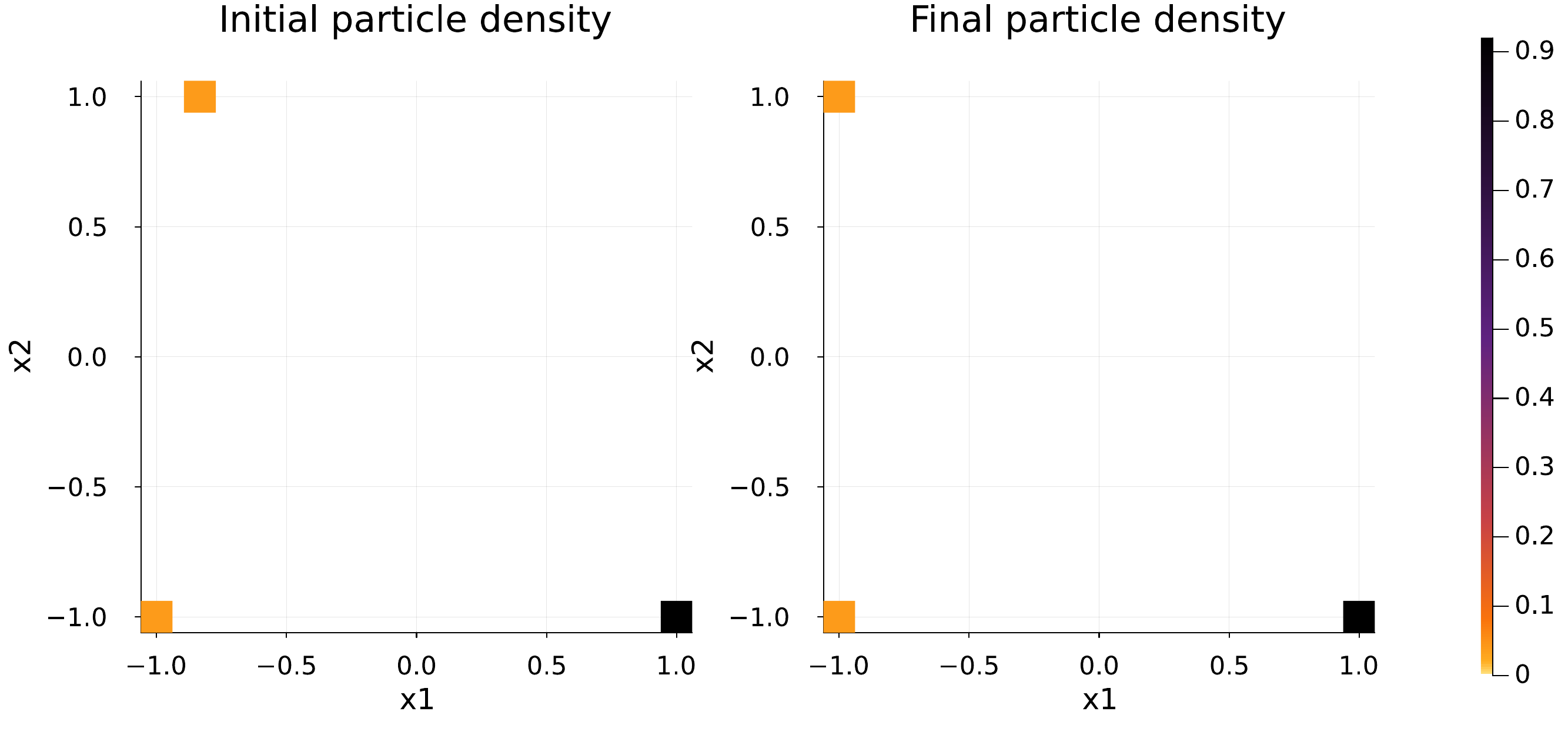}
    \includegraphics[width=0.495\textwidth]{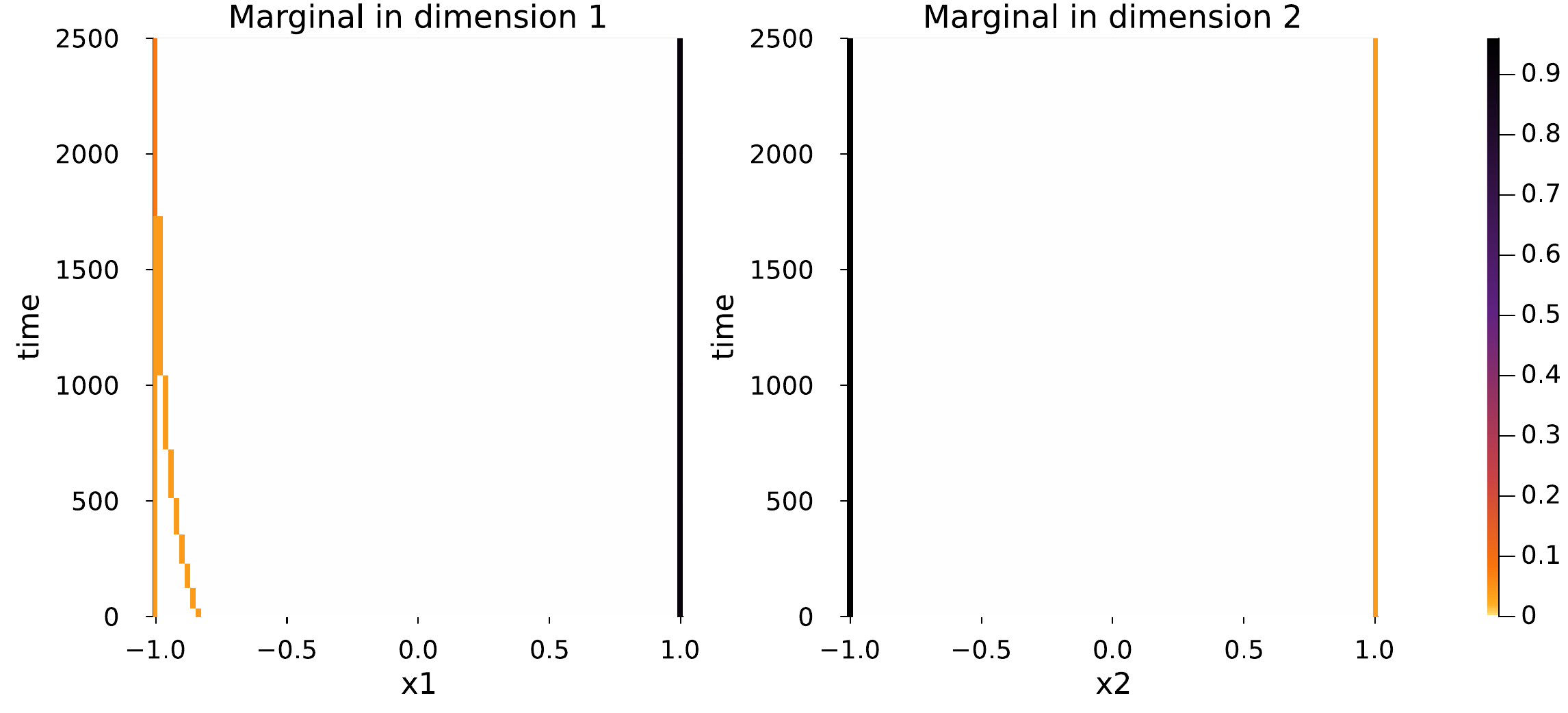}
    \caption{\ref{it:p_v} and \eqref{eq:ex1_new}}
    \label{fig:varalp}
  \end{subfigure}
    \begin{subfigure}[t]{\textwidth}
    \includegraphics[width=0.495\textwidth]{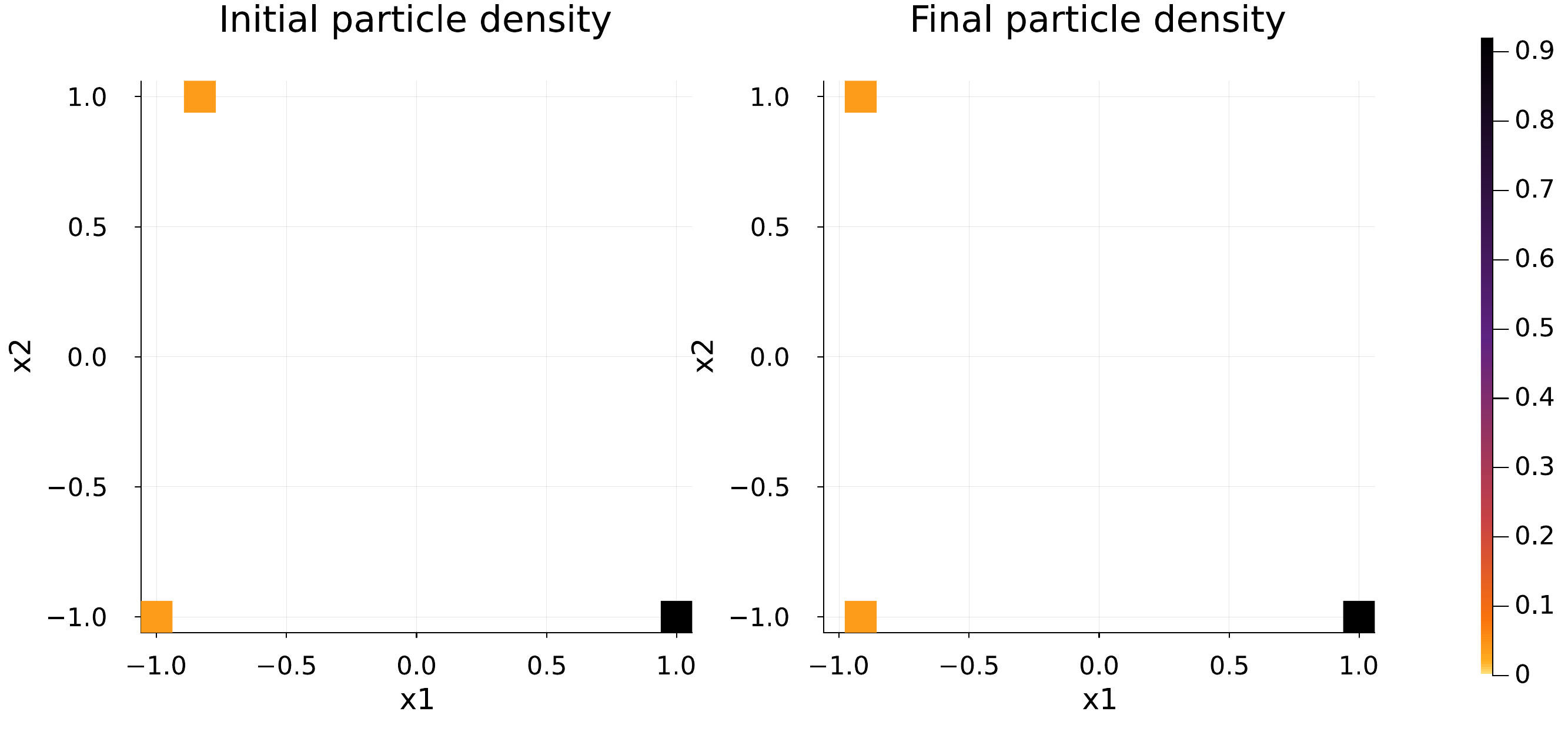}
    \includegraphics[width=0.495\textwidth]{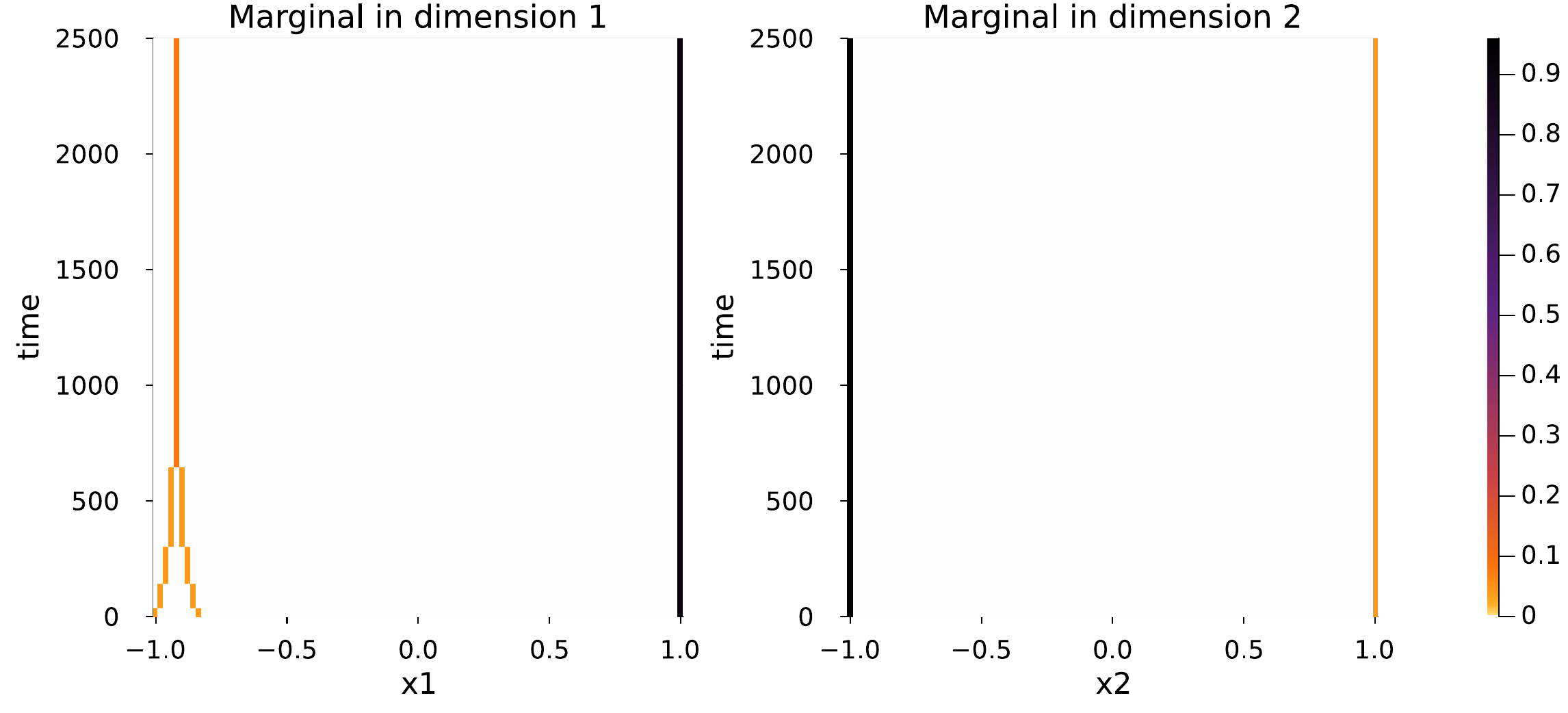}
    \caption{\ref{it:p_f} with $\alpha\equiv(\frac{4}{5},\frac{1}{5})$ and \eqref{eq:ex1_noalp}}
    \label{fig:alp_0.8}
  \end{subfigure}
    \begin{subfigure}[t]{\textwidth}
    \includegraphics[width=0.495\textwidth]{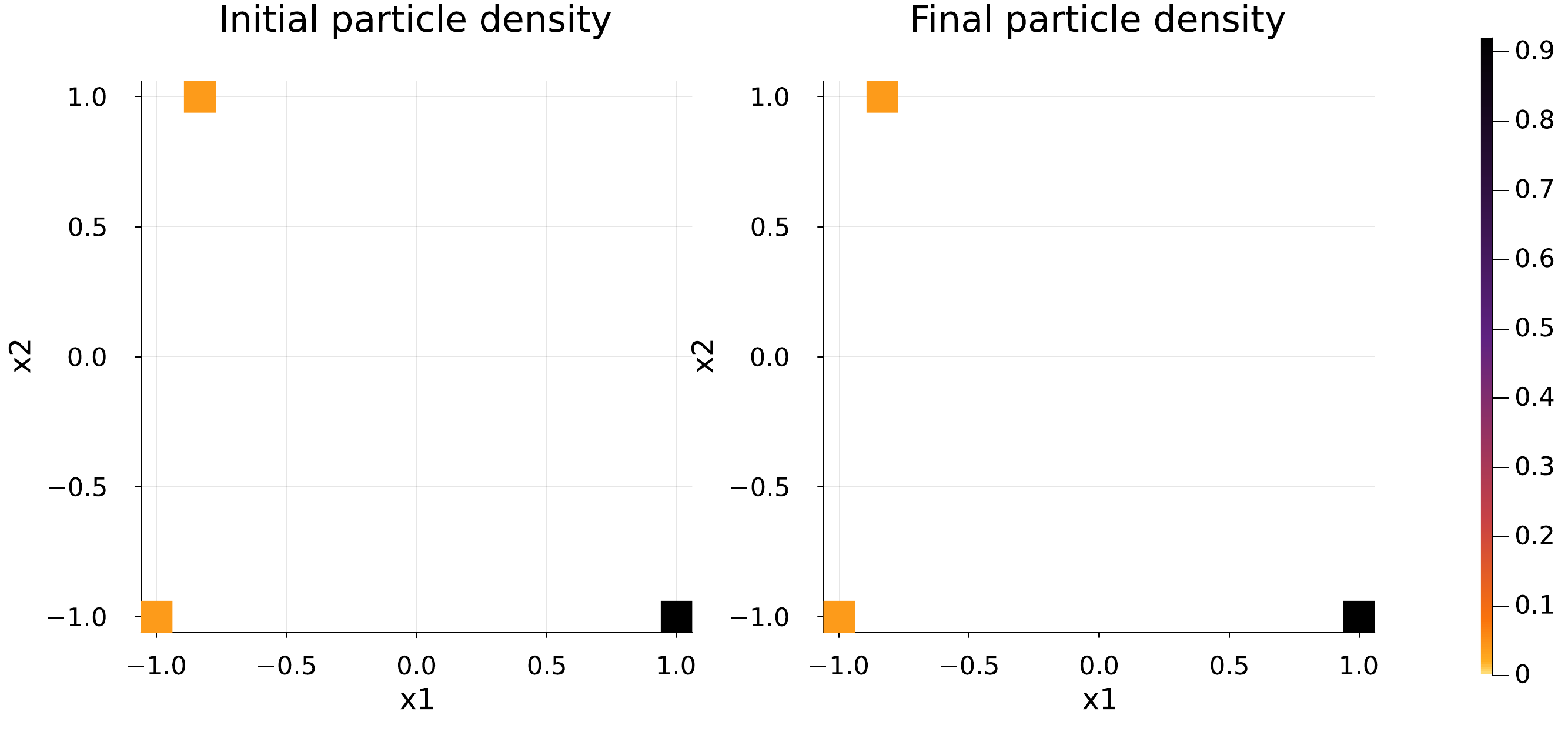}
    \includegraphics[width=0.495\textwidth]{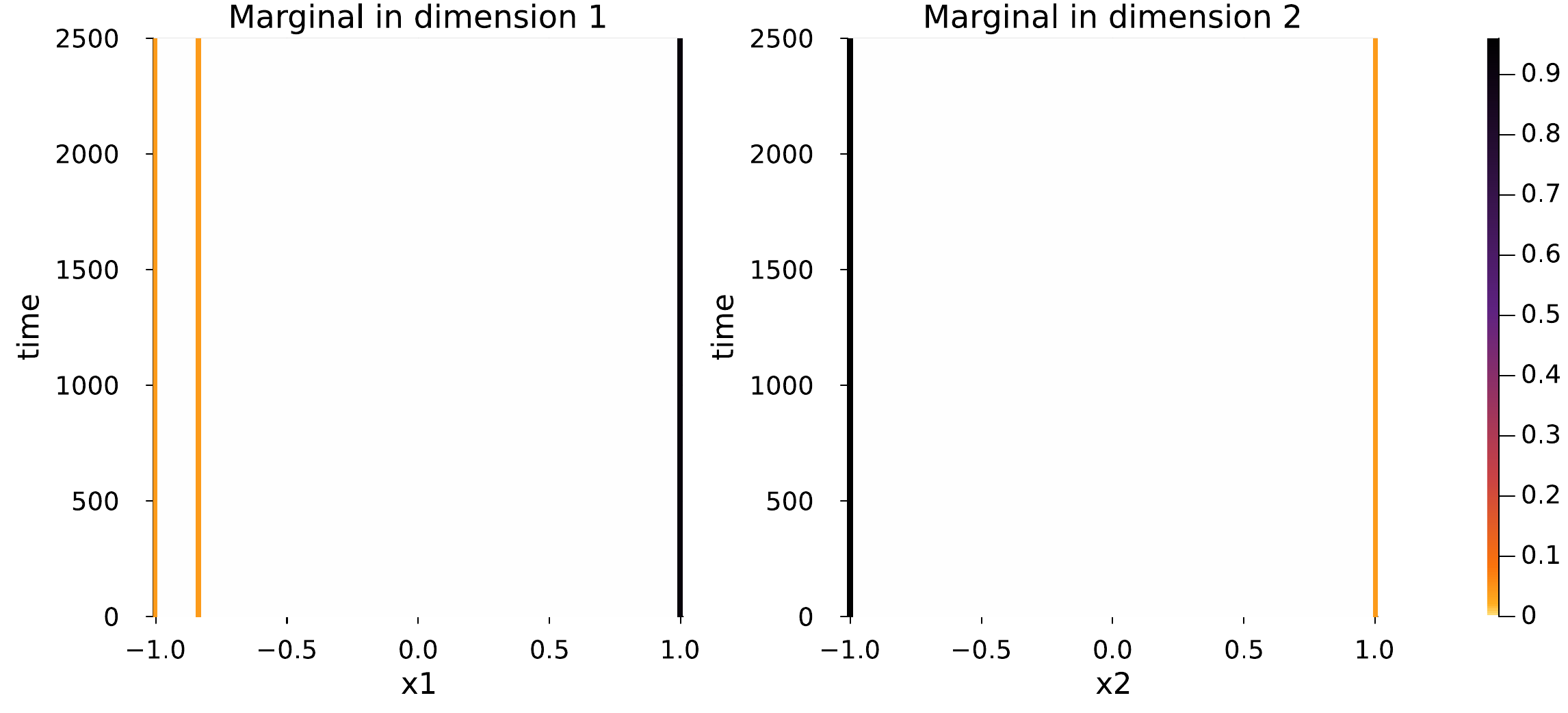}
     \caption{\ref{it:p_f} with $\alpha\equiv(\frac{1}{2},\frac{1}{2})$ and \eqref{eq:ex1_noalp}}
     \label{fig:alp_0.5}
  \end{subfigure}
    \begin{subfigure}[t]{\textwidth}
    \includegraphics[width=0.495\textwidth]{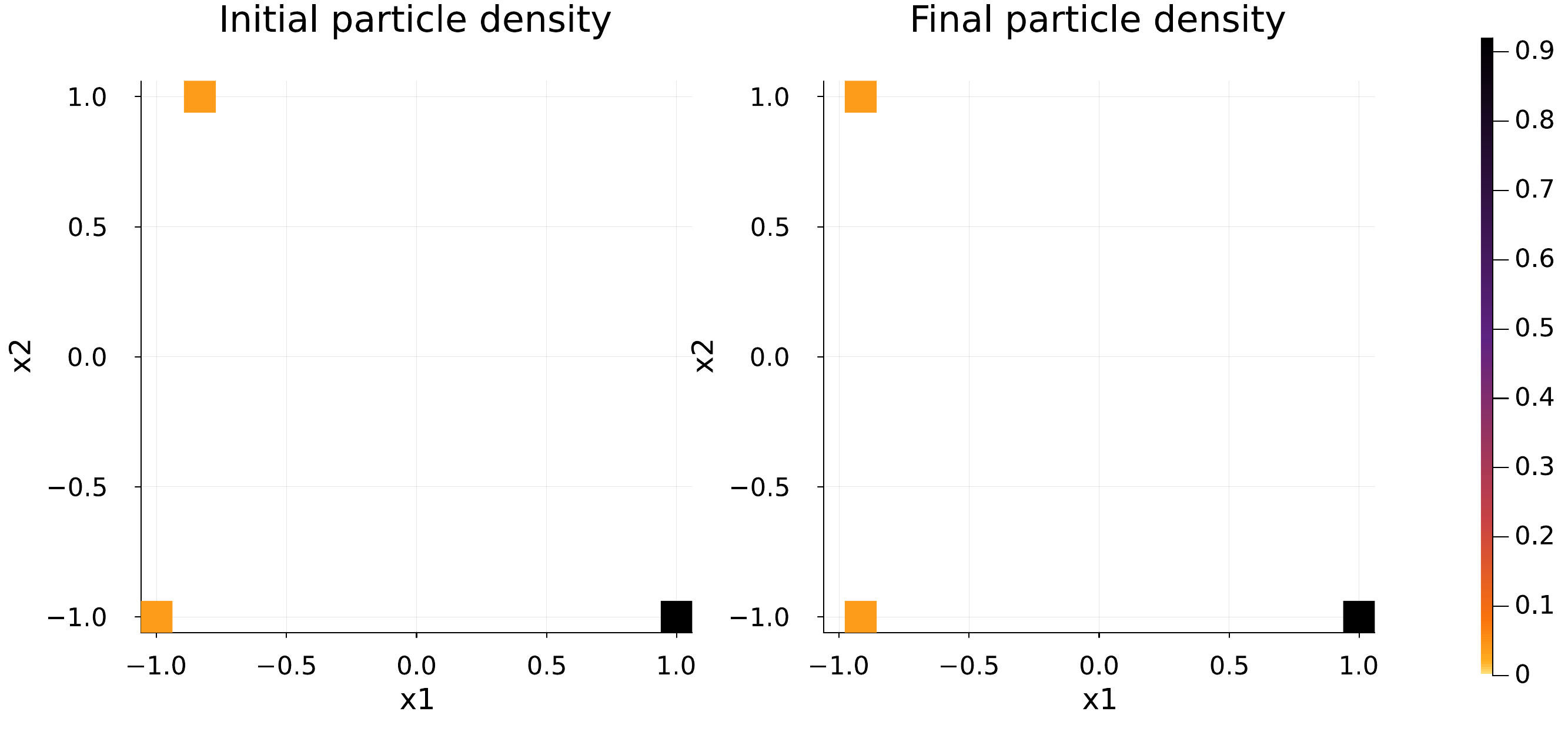}
    \includegraphics[width=0.495\textwidth]{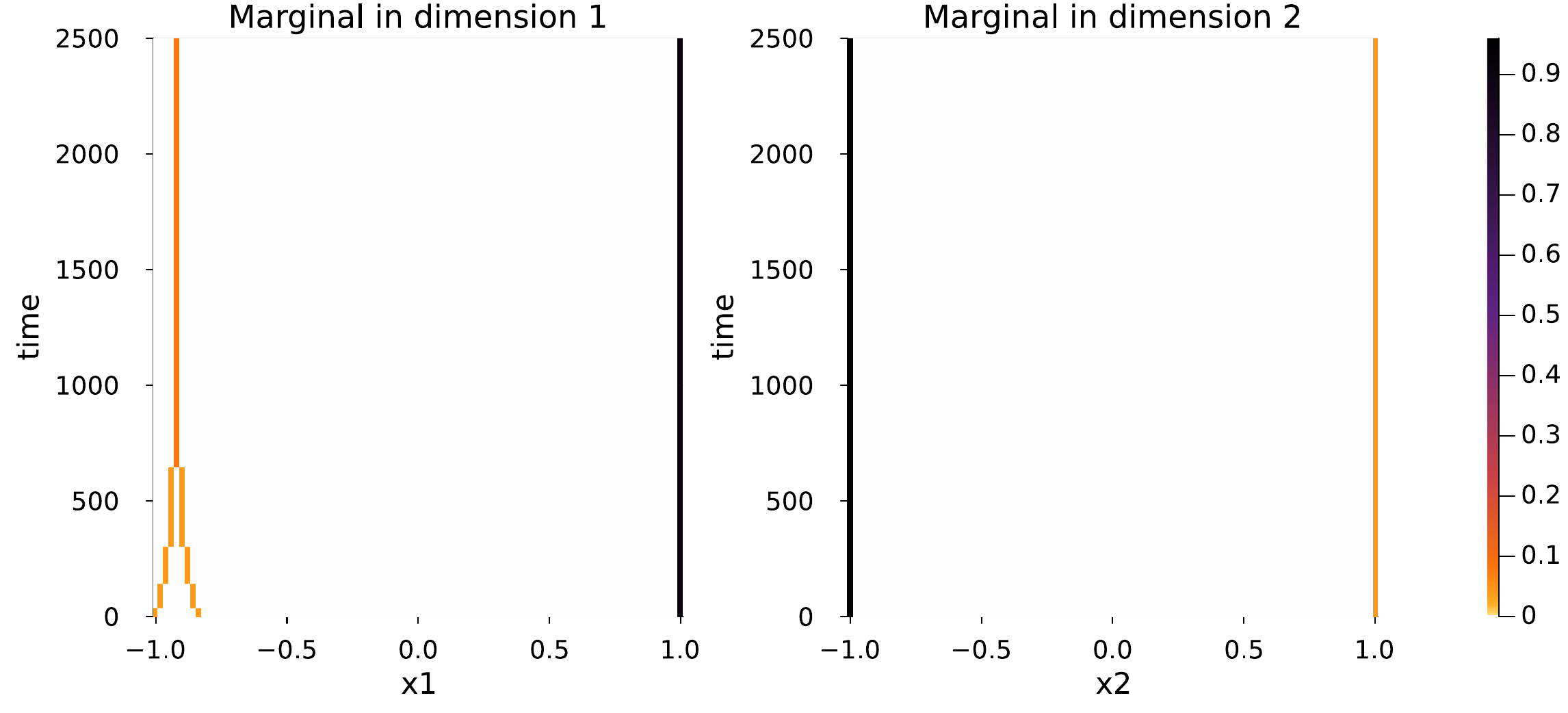}
    \caption{\ref{it:comp_eucl} and \eqref{eq:ex1_noalp}}
    \label{fig:hk}
    \end{subfigure}
    \begin{subfigure}[t]{\textwidth}
    \includegraphics[width=0.495\textwidth]{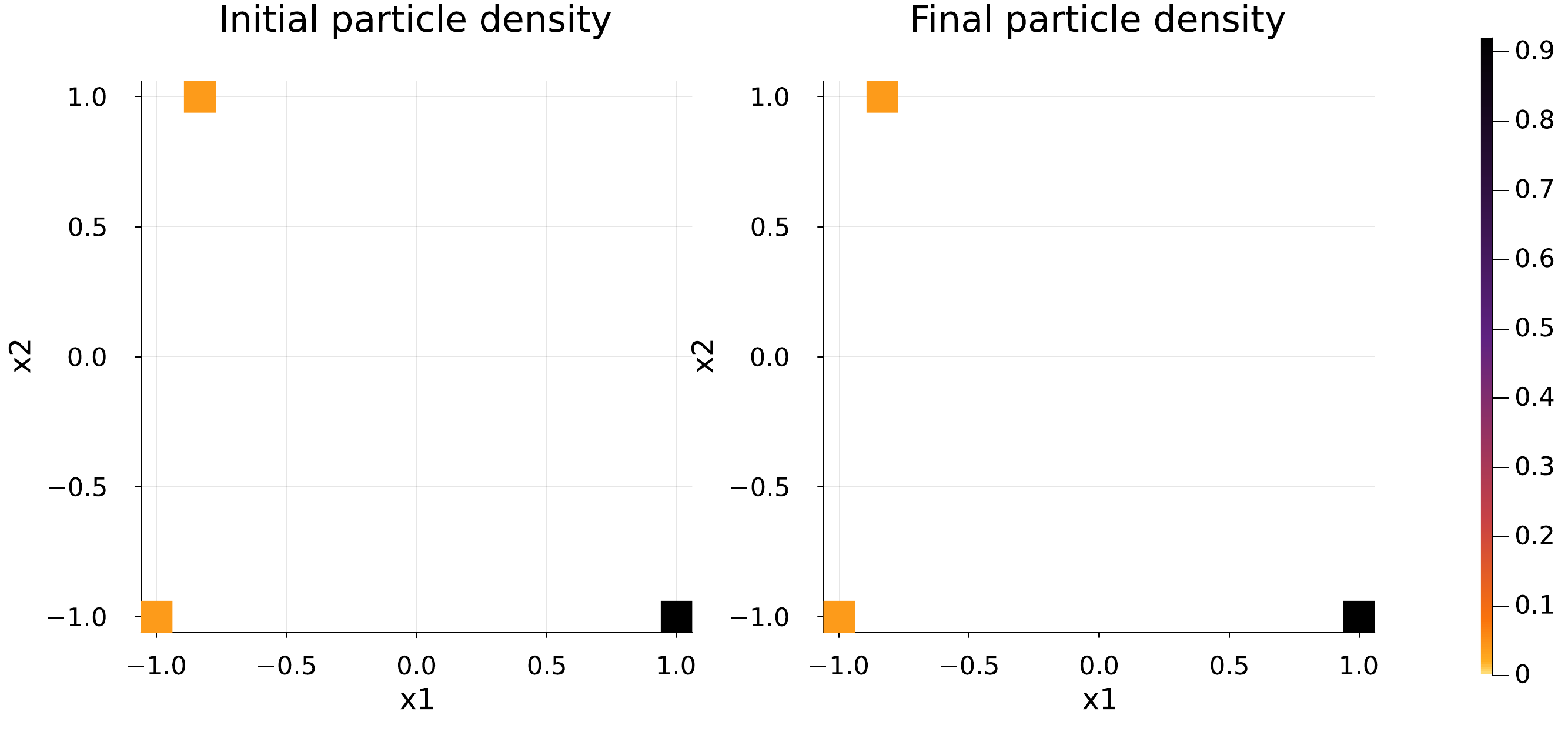}
    \includegraphics[width=0.495\textwidth]{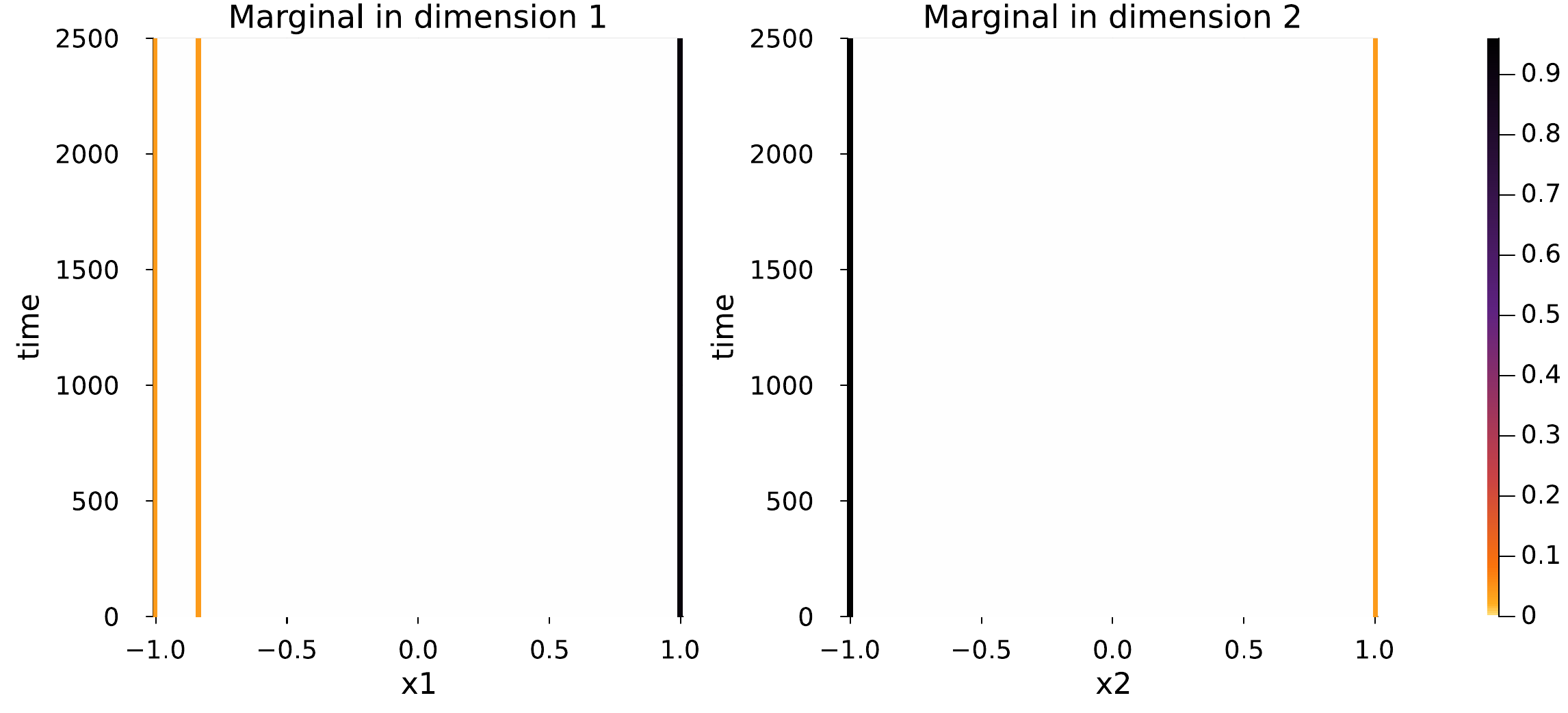}
    \caption{\ref{it:eucl} and \eqref{eq:ex1_noalp}}
    \label{fig:eucli}
  \end{subfigure}
\caption{Initial distribution given by \eqref{eq:ex1_new} and the corresponding stationary states illustrating the impact of different distances discussed in Subsection \ref{sec:sim_ex1}}
  \label{fig:difdis}
\end{figure}

\FloatBarrier
\subsection{From left-wing to right-wing}\label{sec:ltr} 

Let us now demonstrate another case that would not be possible to observe without considering the interplay between topics, and that demonstrates the effect of the choice of interaction radius $r_1$. For this we assume that most people have "right-wing" or "left-wing" opinions corresponding to $(\frac{3}{4},\frac{3}{4},\frac{3}{4})$ and $(-\frac{3}{4},-\frac{3}{4},-\frac{3}{4})$ respectively. Those people weigh all topics equally, i.e. $\alpha=(\frac{1}{3},\frac{1}{3},\frac{1}{3})$. We further assume that a few people have one "right-wing" and two "left-wing" opinions, $(\frac{3}{4},-\frac{3}{4},-\frac{3}{4})$, and $\alpha=(\frac{7}{9},\frac{1}{9},\frac{1}{9})$, which means that the first topic is significantly more important to them than the other two topics. We can write that as initial condition \begin{align}
    f_0(x,\alpha)=\frac{9}{20}\delta_{((\frac{3}{4},\frac{3}{4},\frac{3}{4}),(\frac{1}{3},\frac{1}{3},\frac{1}{3}))}(x,\alpha)+\frac{9}{20}\delta_{((-\frac{3}{4},-\frac{3}{4},-\frac{3}{4}),(\frac{1}{3},\frac{1}{3},\frac{1}{3}))}(x,\alpha)+\frac{1}{10}\delta_{((\frac{3}{4},-\frac{3}{4},-\frac{3}{4}),(\frac{7}{9},\frac{1}{9},\frac{1}{9}))}(x,\alpha).
\end{align}
In Table \ref{tab:parameterchoice_3d}, we display the parameter choices we used for the simulations.

\begin{table}[ht]
    \centering
    \begin{tabular}{c c c}
    \toprule
    \textbf{Parameter} & \textbf{Notation} & \textbf{Value} \\
      \midrule
        number of topics & $d$ & $3$ \\[2pt]
         ratio of current vs other topics & $\beta$ & $\frac{1}{2}$\\[2pt]
         lower bound for \eqref{eq:phismooth} & $r_1$ & $\frac{11}{12}$ \\[2pt]
         upper bound for \eqref{eq:phismooth} & $r_2$ & $r_1+0.0001$ \\[2pt]
        final time & $T$ & $700$ \\[2pt]
        scaling factor & $s$ & $25$ \\[2pt]
        number of particles & $N$ & $24$ \\
    \bottomrule    
    \end{tabular}
    \caption{Parameters used in Subsection \ref{sec:ltr}}
    \label{tab:parameterchoice_3d}
\end{table}

As we can see in Figure \ref{fig:initfin_3d_init_exact}, the people having one "right-wing" and two "left-wing" opinions at the beginning of the simulation, have three "right-wing" opinions at the end of the simulation which they share with the people who already had three "right-wing" opinions at the start. This is a behaviour that occurs because of the way we choose $\alpha$. Furthermore, it only happens because the opinions on different topic are related and people have different $\alpha$s.  This dynamic can also be seen in Figure \ref{fig:marginal_3d_init_exact} where the marginals in each opinion are plotted over time.
\begin{figure}[!ht]
  \centering
	\includegraphics[width=1\linewidth]{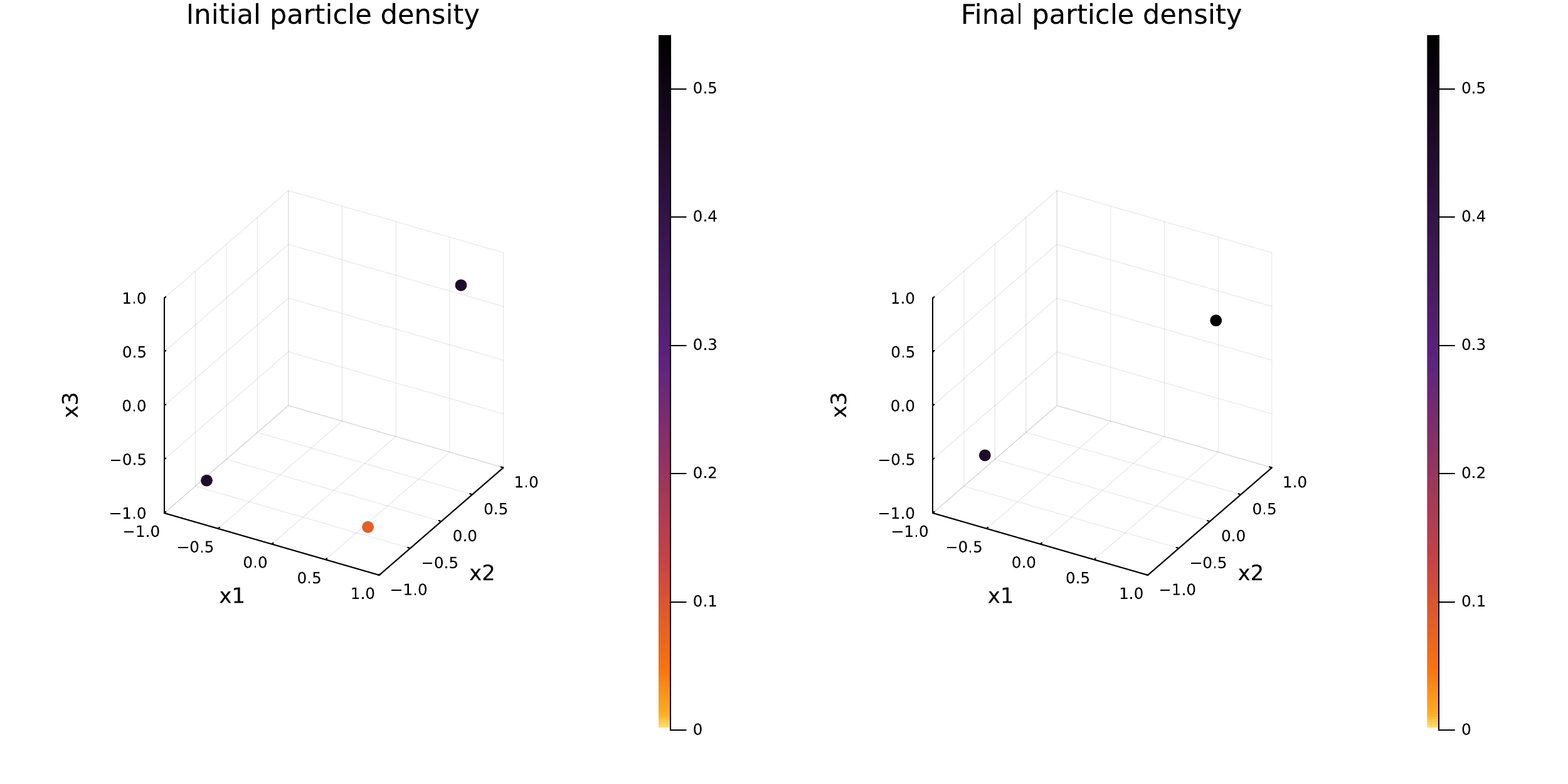}
    \caption{Initial and final particle density illustrating the swing from 'left' to 'right' discussed in Section \ref{sec:ltr}}
    \label{fig:initfin_3d_init_exact}
\end{figure}

The choice of the interaction radius, in particular $r_1$, plays a significant role regarding what behaviour can be observed. This can be seen in Figure \ref{fig:diff_r_1}, where at the final time step we can observe 3 clusters in Figure \ref{fig:marginal_3d_init_exact_nointer}, 2 clusters in Figure \ref{fig:marginal_3d_init_exact}, 2 clusters and consensus regarding the 2nd and 3rd topic in Figure \ref{fig:marginal_3d_init_exact_consensusin2top}, or consensus in Figure \ref{fig:marginal_3d_init_exact_consensus}, depending on the choice of $r_1$. This shows that, as we would expect, the bigger the interaction radius the more interactions are happening.

\begin{figure}[!ht]  
  \begin{subfigure}[t]{\textwidth}
    \includegraphics[width=\textwidth]{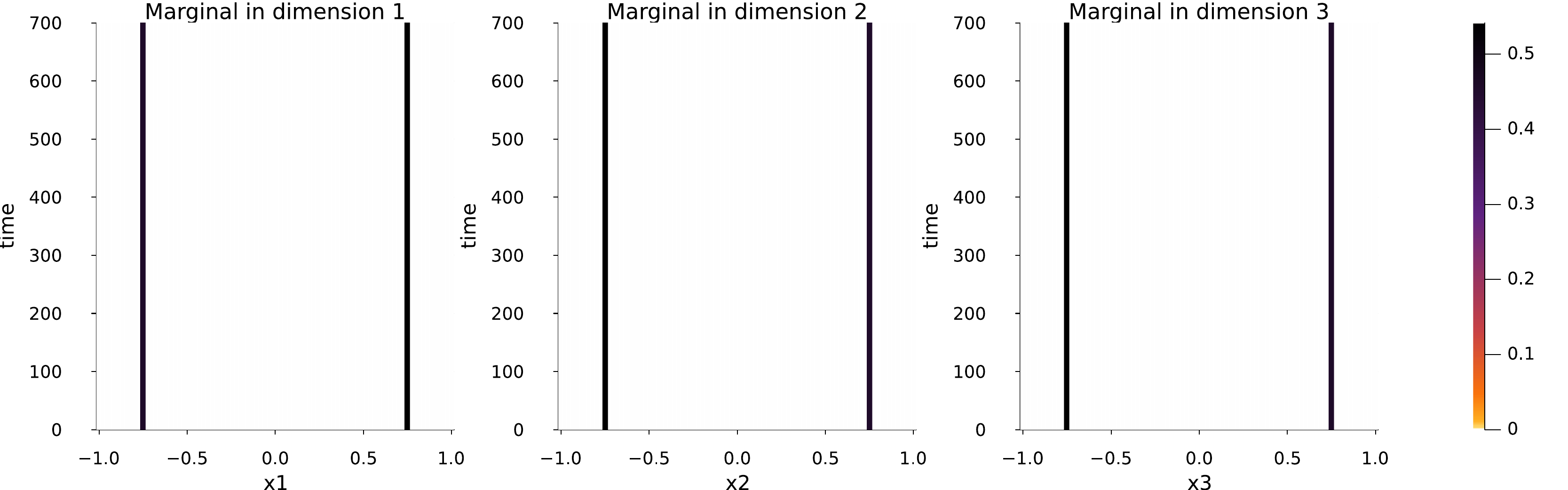}
    \caption{$r_1=0.9$}
    \label{fig:marginal_3d_init_exact_nointer}
  \end{subfigure}
    \begin{subfigure}[t]{\textwidth}
    \includegraphics[width=\textwidth]{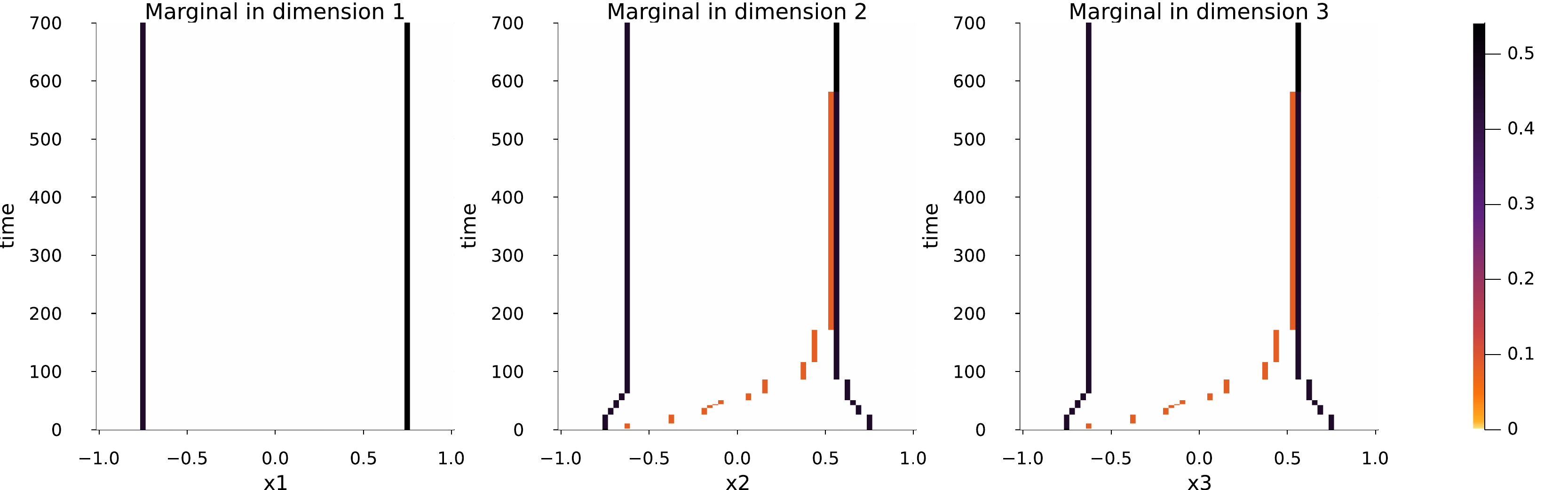}
    \caption{$r_1=\frac{11}{12}$}
    \label{fig:marginal_3d_init_exact}
  \end{subfigure}
  \begin{subfigure}[t]{\textwidth}
    \includegraphics[width=\textwidth]{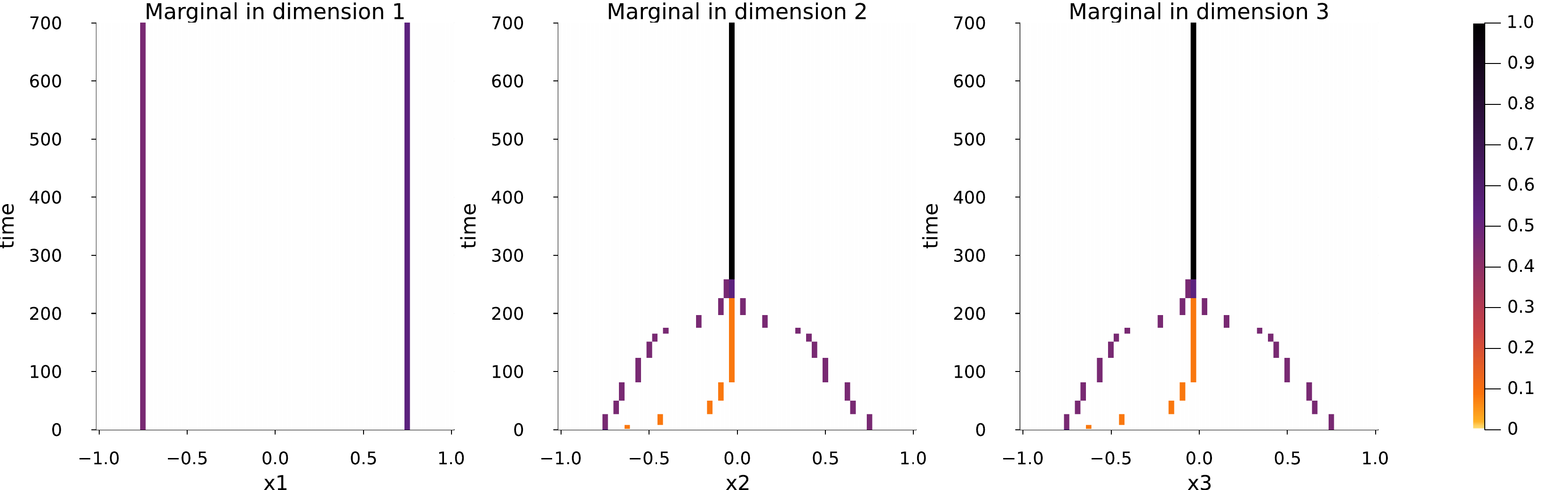}
    \caption{$r_1=0.99$}
    \label{fig:marginal_3d_init_exact_consensusin2top}
  \end{subfigure}
    \begin{subfigure}[t]{\textwidth}
    \includegraphics[width=\textwidth]{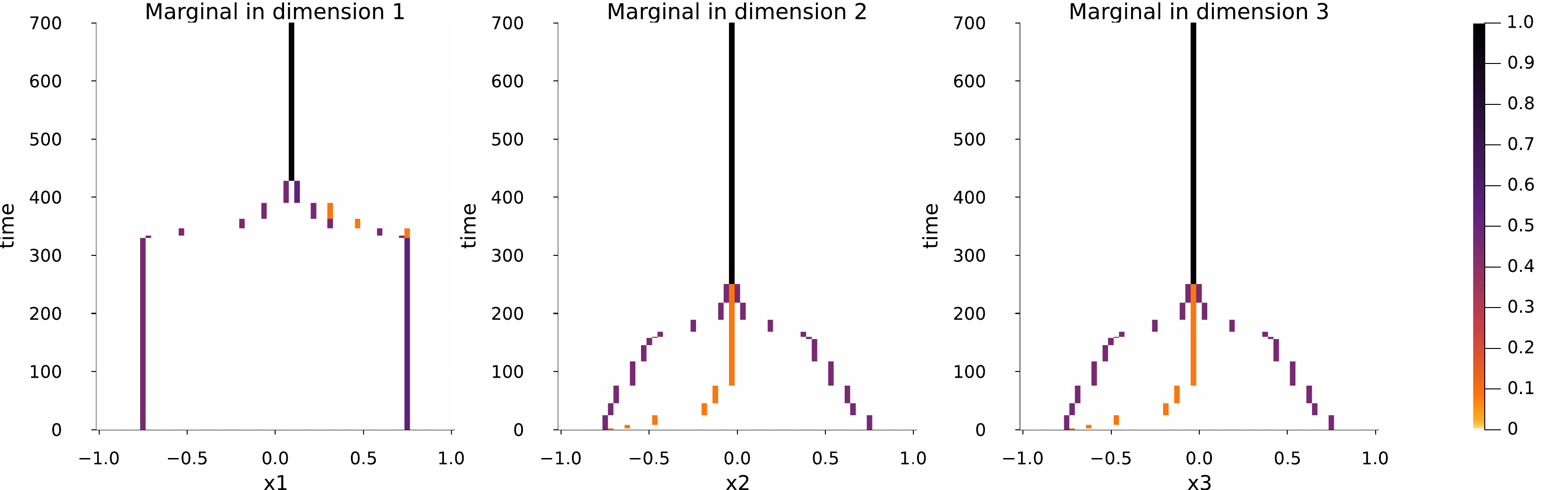}
    \caption{$r_1=1$}
    \label{fig:marginal_3d_init_exact_consensus}
  \end{subfigure}
    \caption{Marginal plots of solutions to \eqref{eq:pde} for different values of $r_1$ showing the influence of the interaction radii on the number of clusters}
    \label{fig:diff_r_1}
\end{figure}

\section{Conclusion and future work}\label{sec:conclusion}
In this paper, we introduced a model for multi-dimensional opinion dynamics for connected topics. People change their opinion on each topic, based on their distance in opinion - this distance  depends on individual importance weights of different topics. We first consider a kinetic formulation of the model, from which we derive the respective PDE in the mean field limit. Then we showed some analytic properties and convergence results for particular cases. Moreover, we demonstrated that due to the individual importance weights, the average opinion vector can change and the variance can increase. This dynamics can only be observed in case of individual important weights and differs from other proposed distances like the Euclidean distance. \\

\noindent Future work includes the convergence to steady state in case of different importance weights, as well as the full characterisation of stationary states. Another possible research direction corresponds to opinion control by influencing individual opinion weights.

\section{Acknowledgements}
HB acknowledges the support by the Deutsche Forschungsgemeinschaft (DFG, 
German Research Foundation) within the Research Training Group GRK 2583 
''Modeling, Simulation and Optimisation of Fluid Dynamic Applications''.
MB acknowledges support from DESY (Hamburg, Germany), a member of the Helmholtz Association HGF.

\printbibliography 

@article{doi:10.1137/25M1730818,
author = {Li, Grace Jingying and Luo, Jiajie and Chu, Weiqi},
title = {Bounded-Confidence Models of Multidimensional Opinions with Topic-Weighted Discordance},
journal = {SIAM Journal on Applied Dynamical Systems},
volume = {25},
number = {1},
pages = {1-41},
year = {2026},
doi = {10.1137/25M1730818},
URL = {https://doi.org/10.1137/25M1730818},
eprint = {https://doi.org/10.1137/25M1730818}
}

@article{diperna1989ordinary,
  title={Ordinary differential equations, transport theory and Sobolev spaces},
  author={DiPerna, Ronald J and Lions, Pierre-Louis},
  journal={Inventiones mathematicae},
  volume={98},
  number={3},
  pages={511--547},
  year={1989},
  publisher={Springer}
}

@phdthesis{AN2025,
  author  = "Andrew Nugent",
  title   = "Convergence, Control and Continuous Ageing
in Opinion Dynamics",
  school  = "University of Warwick",
  year    = "2025"
}

@INPROCEEDINGS{nedic_multidhk,
  author={Nedić, A. and Touri, B.},
  booktitle={2012 IEEE 51st IEEE Conference on Decision and Control (CDC)}, 
  title={Multi-dimensional Hegselmann-Krause dynamics}, 
  year={2012},
  volume={},
  number={},
  pages={68-73},
  doi={10.1109/CDC.2012.6426417}}

@article{bayraktar2023graphon,
  title={Graphon mean field systems},
  author={Bayraktar, Erhan and Chakraborty, Suman and Wu, Ruoyu},
  journal={The Annals of Applied Probability},
  volume={33},
  number={5},
  pages={3587--3619},
  year={2023},
  publisher={Institute of Mathematical Statistics}
}

@article{during2024breaking,
  title={Breaking consensus in kinetic opinion formation models on graphons},
  author={D{\"u}ring, Bertram and Franceschi, Jonathan and Wolfram, Marie-Therese and Zanella, Mattia},
  journal={Journal of Nonlinear Science},
  volume={34},
  number={4},
  pages={79},
  year={2024},
  publisher={Springer}
}

@article{bondesan2024kinetic,
  title={Kinetic compartmental models driven by opinion dynamics: vaccine hesitancy and social influence},
  author={Bondesan, Andrea and Toscani, Giuseppe and Zanella, Mattia},
  journal={Mathematical Models and Methods in Applied Sciences},
  volume={34},
  number={06},
  pages={1043--1076},
  year={2024},
  publisher={World Scientific}
}

@article{degroot1974reaching,
  title={Reaching a consensus},
  author={DeGroot, Morris H},
  journal={Journal of the American Statistical association},
  volume={69},
  number={345},
  pages={118--121},
  year={1974},
  publisher={Taylor \& Francis}
}

@article{deffuant2000mixing,
  title={Mixing beliefs among interacting agents},
  author={Deffuant, Guillaume and Neau, David and Amblard, Frederic and Weisbuch, G{\'e}rard},
  journal={Advances in Complex Systems},
  volume={3},
  number={01n04},
  pages={87--98},
  year={2000},
  publisher={World Scientific}
}

@article{rodriguez_collective_2016,
	title = {Collective {Dynamics} of {Belief} {Evolution} under {Cognitive} {Coherence} and {Social} {Conformity}},
	volume = {11},
	issn = {1932-6203},
	url = {https://dx.plos.org/10.1371/journal.pone.0165910},
	doi = {10.1371/journal.pone.0165910},
	language = {en},
	number = {11},
	urldate = {2024-10-07},
	journal = {PLOS ONE},
	author = {Rodriguez, Nathaniel and Bollen, Johan and Ahn, Yong-Yeol},
	editor = {Lambiotte, Renaud},
	month = nov,
	year = {2016},
	pages = {e0165910},
}

@article{schweighofer_agent-based_2020,
	title = {An agent-based model of multi-dimensional opinion dynamics and opinion alignment},
	volume = {30},
	issn = {1054-1500},
	url = {https://doi.org/10.1063/5.0007523},
	doi = {10.1063/5.0007523},
	number = {9},
	urldate = {2025-10-07},
	journal = {Chaos: An Interdisciplinary Journal of Nonlinear Science},
	author = {Schweighofer, Simon and Garcia, David and Schweitzer, Frank},
	month = sep,
	year = {2020},
	pages = {093139},
}

@article{schweighofer_weighted_2020,
	title = {A {Weighted} {Balance} {Model} of {Opinion} {Hyperpolarization}},
	volume = {23},
	issn = {1460-7425},
	number = {3},
	journal = {Journal of Artificial Societies and Social Simulation},
	author = {Schweighofer, Simon and Schweitzer, Frank and Garcia, David},
	year = {2020},
	pages = {5},
}

@article{ojer_social_2025,
	title = {Social network heterogeneity promotes depolarization of multidimensional correlated opinions},
	volume = {7},
	url = {https://link.aps.org/doi/10.1103/PhysRevResearch.7.013207},
	doi = {10.1103/PhysRevResearch.7.013207},
	number = {1},
	urldate = {2025-10-07},
	journal = {Physical Review Research},
	author = {Ojer, Jaume and Starnini, Michele and Pastor-Satorras, Romualdo},
	month = feb,
	year = {2025},
	note = {Publisher: American Physical Society},
	pages = {013207},
}

@book{kelley_theory_2010,
	address = {New York, NY},
	title = {The {Theory} of {Differential} {Equations}: {Classical} and {Qualitative}},
	copyright = {https://www.springernature.com/gp/researchers/text-and-data-mining},
	isbn = {978-1-4419-5782-5 978-1-4419-5783-2},
	shorttitle = {The {Theory} of {Differential} {Equations}},
	url = {https://link.springer.com/10.1007/978-1-4419-5783-2},
	language = {en},
	urldate = {2025-09-05},
	publisher = {Springer New York},
	author = {Kelley, Walter G. and Peterson, Allan C.},
	year = {2010},
	doi = {10.1007/978-1-4419-5783-2},
}

@online{ODESolver,
title = {ODE Solvers - DifferentialEquations.jl},
url = {https://docs.sciml.ai/DiffEqDocs/stable/solvers/ode_solve/},
urldate = {2025-05-28}
}

@article{hegselmannkrause,
author = {Hegselmann, Rainer and Krause, Ulrich},
year = {2002},
month = {07},
pages = {},
title = {Opinion Dynamics and Bounded Confidence Models, Analysis and Simulation},
volume = {5},
journal = {Journal of Artificial Societies and Social Simulation},
URL = {https://jasss.soc.surrey.ac.uk/5/3/2.html}
}

@article{cahill_modified_2025,
	title = {A modified {Hegselmann}-{Krause} model for interacting voters and political parties},
	url = {http://arxiv.org/abs/2410.13378},
	doi = {10.48550/arXiv.2410.13378},
	urldate = {2025-03-03},
	publisher = {arXiv},
	author = {Cahill, Patrick H. and Gottwald, Georg A.},
	month = feb,
	year = {2025},
	note = {arXiv:2410.13378 [physics]},
}

@article{fortunato_vector_2005,
	title = {Vector {Opinion} {Dynamics} in a {Bounded} {Confidence} {Consensus} {Model}},
	volume = {16},
	issn = {0129-1831, 1793-6586},
	url = {http://arxiv.org/abs/physics/0504017},
	doi = {10.1142/S0129183105008126},
	number = {10},
	urldate = {2024-10-07},
	journal = {International Journal of Modern Physics C},
	author = {Fortunato, Santo and Latora, Vito and Pluchino, Alessandro and Rapisarda, Andrea},
	month = oct,
	year = {2005},
	note = {arXiv:physics/0504017},
	pages = {1535--1551},
}

@article{pedraza_analytical_2021,
	title = {An analytical formulation for multidimensional continuous opinion models},
	volume = {152},
	issn = {0960-0779},
	url = {https://www.sciencedirect.com/science/article/pii/S0960077921007220},
	doi = {10.1016/j.chaos.2021.111368},
	urldate = {2024-10-16},
	journal = {Chaos, Solitons \& Fractals},
	author = {Pedraza, Lucía and Pinasco, Juan Pablo and Saintier, Nicolas and Balenzuela, Pablo},
	month = nov,
	year = {2021},
	pages = {111368},
}

@article{noipitak_dynamics_2021,
	title = {Dynamics of interdependent multidimensional opinions},
	volume = {1719},
	issn = {1742-6588, 1742-6596},
	url = {https://iopscience.iop.org/article/10.1088/1742-6596/1719/1/012107},
	doi = {10.1088/1742-6596/1719/1/012107},
	language = {en},
	number = {1},
	urldate = {2024-10-17},
	journal = {Journal of Physics: Conference Series},
	author = {Noipitak, S and Allen, M A},
	month = jan,
	year = {2021},
	pages = {012107},
}

@article{nugent_evolving_2023,
	title = {On evolving network models and their influence on opinion formation},
	volume = {456},
	issn = {0167-2789},
	url = {https://www.sciencedirect.com/science/article/pii/S0167278923002683},
	doi = {10.1016/j.physd.2023.133914},
	urldate = {2024-10-22},
	journal = {Physica D: Nonlinear Phenomena},
	author = {Nugent, Andrew and Gomes, Susana N. and Wolfram, Marie-Therese},
	month = dec,
	year = {2023},
	pages = {133914},
}

@article{noorazar_classical_2020,
	title = {From classical to modern opinion dynamics},
	volume = {31},
	issn = {0129-1831, 1793-6586},
	url = {http://arxiv.org/abs/1909.12089},
	doi = {10.1142/S0129183120501016},
	language = {en},
	number = {07},
	urldate = {2024-10-23},
	journal = {International Journal of Modern Physics C},
	author = {Noorazar, Hossein and Vixie, Kevin R. and Talebanpour, Arghavan and Hu, Yunfeng},
	month = jul,
	year = {2020},
	note = {arXiv:1909.12089 [physics]},
	pages = {2050101},
}

@article{toscani_kinetic_2006,
author = {Giuseppe Toscani},
title = {{Kinetic models of opinion formation}},
volume = {4},
journal = {Communications in Mathematical Sciences},
number = {3},
publisher = {International Press of Boston},
pages = {481 -- 496},
year = {2006},
}

@article{albi_opinion_2016,
	title = {Opinion dynamics over complex networks: {Kinetic} modelling and numerical methods},
	volume = {10},
	copyright = {http://creativecommons.org/licenses/by/3.0/},
	issn = {1937-5093},
	shorttitle = {Opinion dynamics over complex networks},
	url = {https://www.aimsciences.org/en/article/doi/10.3934/krm.2017001},
	doi = {10.3934/krm.2017001},
	language = {en},
	number = {1},
	urldate = {2024-11-08},
	journal = {Kinetic and Related Models},
	author = {Albi, Giacomo and Pareschi, Lorenzo and Zanella, Mattia},
	month = nov,
	year = {2016},
	note = {Publisher: Kinetic and Related Models},
	pages = {1--32},
}

@article{boudin_kinetic_2009,
	title = {A kinetic approach to the study of opinion formation},
	volume = {43},
	issn = {1290-3841},
	url = {http://www.numdam.org/item/M2AN_2009__43_3_507_0/},
	doi = {10.1051/m2an/2009004},
	language = {en},
	number = {3},
	urldate = {2024-11-12},
	journal = {ESAIM: Modélisation mathématique et analyse numérique},
	author = {Boudin, Laurent and Salvarani, Francesco},
	year = {2009},
	pages = {507--522},
}

@article{during_boltzmann_2009,
	title = {Boltzmann and {Fokker}–{Planck} equations modelling opinion formation in the presence of strong leaders},
	volume = {465},
	issn = {1364-5021, 1471-2946},
	url = {https://royalsocietypublishing.org/doi/10.1098/rspa.2009.0239},
	doi = {10.1098/rspa.2009.0239},
	language = {en},
	number = {2112},
	urldate = {2024-11-13},
	journal = {Proceedings of the Royal Society A: Mathematical, Physical and Engineering Sciences},
	author = {Düring, Bertram and Markowich, Peter and Pietschmann, Jan-Frederik and Wolfram, Marie-Therese},
	month = dec,
	year = {2009},
	pages = {3687--3708},
}

\appendix
\section{Derivation of the mean-field PDE if people interact in a single topic per interaction}\label{app:onlyonetopicperinter}

In the following we present the derivation of the strong PDE if people only interact on a single topic per interaction. Let $\zeta$ be a discrete uniformly distributed random variable that takes values in $\{1,\dots,d\}$. Then we define the respective binary interaction between as
\begin{align}\label{eq:interactions_1topic}
\begin{split}
&\text{If }\zeta = a \text{ then }\\
& \qquad x_a^{*_a} = x_a + \gamma \phi(p_a(x,y;\alpha))(y_a-x_a)  \\
& \qquad y_a^{*_a} = y_a + \gamma \phi(p_a(y,x;\mu) (x_a-y_a)\\
&\text{For all other } b\in \DD\backslash\{a\}\\
& \qquad x_b^{*_a} = x_b, y_b^{*_a} = y_b
\end{split}
\end{align}
where the pre-interaction opinion and weights are given by $(x,\alpha), (y,\mu) \in \mR^d\times\Omega$. Again, the parameter $\gamma\in (0,1)$ states how strong the influence of interaction on the opinion is.

Next we wish to compute the respective mean-field limit. In doing so, we follow the same approach as in Section \ref{sec:derivation} and consider
\begin{align}\label{eq:dtglalp_1topic}
\frac{\partial f(x,\alpha,t)}{\partial t} = 
G(x,\alpha) - L(x,\alpha).
\end{align}

To define the gain and the loss term, we assume the same considerations as before to obtain

\begin{align*}
    G&(f,f)(x', \alpha,t) =\\
    &=\kappa \intiai \sum_{a=1}^d\int\limits_{\DD} \delta\left(w-a\right)\delta\left(x' - x^{*_a}(x,y,\alpha)\right) \mathds{1}_{\lbrace x\neq x^{*_a}\rbrace}f(x,\alpha,t) f(y, \eta, t)  \dd \zeta(w) \dd x \dd (y,\eta)\\
    &=\frac{\kappa}{d} \intiai \sum_{a=1}^d \delta\left(x' - x^{*_a}(x,y,\alpha)\right) \mathds{1}_{\lbrace x\neq x^{*_a}\rbrace}f(x,\alpha,t) f(y, \eta, t)   \dd x \dd (y,\eta).
\end{align*}
Here inequality means, that the two vectors differ in at least one component. Hence we exclude interactions which do not alter the original opinion through the indicator function. The first Dirac delta ensures interactions in the $a$-th component only, the second the modified binary interactions.

For the loss term, we obtain (using similar considerations as above)
\begin{align*}
    L(f,f) (x', \alpha,t) =& \kappa \intiai \sum_{a=1}^d\int\limits_{\DD} \delta\left(w-a\right)\delta\left(x' - x\right) \mathds{1}_{\lbrace x\neq x^{*_a}\rbrace }f(x,\alpha,t) f(y, \eta, t) \dd \zeta(w) \dd x \dd (y,\eta)\\
    =& \frac{\kappa}{d} \intiai \sum_{a=1}^d\delta\left(x' - x\right) \mathds{1}_{\lbrace x\neq x^{*_a}\rbrace}f(x,\alpha,t) f(y, \eta, t \dd x \dd (y,\eta).
\end{align*}

Like before, we see that the right hand side can be written as a collision operator in weak form. Let $\xi(x',\alpha)$ be a suitable test function, then 
\begin{align*}
\intia \Bigl(G(f,f) - L(f,f)\Bigr) \xi(x', \alpha)\dd(x', \alpha) = \frac{\kappa}{d}\intiaia\sum_{a=1}^d\Bigl(\xi(x^{*_a}, \alpha)-\xi(x,\alpha)\Bigr) f(x,\alpha,t) f(y, \eta, t) \dd (x,\alpha)\dd(y, \eta).
\end{align*}
Again, without loss of generality, we assume that $\kappa=1$.
To write \eqref{eq:dtgl} in weak form we consider a suitably chosen test function $\xi(x', \alpha)$, denote by $e_a$ the canonical basis vector and obtain
\begin{align*}
    \frac{\dd}{\dd t}&\intia \xi(x,\alpha) f(x,\alpha,t) \dd (x,\alpha)\\ =& \intia \xi(x,\alpha)G(x,\alpha) \dd (x,\alpha) - \intia \xi(x,\alpha) L(x,\alpha) \dd (x,\alpha)\\
    =&\frac{1}{d}\intiaia\sum_{a=1}^d\Bigl(\xi(x^{*_a}, \alpha)-\xi(x,\alpha)\Bigr) f(x,\alpha,t) f(y, \eta, t) \dd (x,\alpha)\dd(y, \eta)\\
    =&\frac{1}{d}\intiaia \sum_{a=1}^d\Bigl(\xi(x+\gamma\phi(p_a(x,y;\alpha))(y_a-x_a)e_a,\alpha)-\xi(x,\alpha)\Bigr) f(x,\alpha,t)f(y,\eta,t) \dd (y,\eta) \dd (x,\alpha) \\
    =&\frac{1}{d} \intiaia  \sum_{a=1}^d\Bigl(\xi(x,\alpha) +\gamma\frac{\dd \xi(x,\alpha)}{\dd x_a} (\phi(p_a(x,y;\alpha))(y_a-x_a)) + \mathcal{O}(\gamma^2)-\xi(x,\alpha)\Bigr) f(x,\alpha,t)f(y,\eta,t) \dd (y,\eta) \dd (x,\alpha)\\
    =& \frac{1}{d}\intiaia  \sum_{a=1}^d\Bigl(\gamma\frac{\dd \xi(x,\alpha)}{\dd x_a} (\phi(p_a(x,y;\alpha))(y_a-x_a)) + \mathcal{O}(\gamma^2)\Bigr) f(x,\alpha,t)f(y,\eta,t) \dd (y,\eta) \dd (x,\alpha)\\
    =& \frac{1}{d}\intiaia  \left(\gamma\nabla_x \xi(x,\alpha) \cdot \left(\phi\left(\left(
\begin{array}{c}
p_1(x,y;\alpha)\\
p_2(x,y;\alpha)\\
\vdots\\
p_d(x,y;\alpha)\\
\end{array}
\right)\right)\odot(y-x)\right) + \mathcal{O}(\gamma^2)\right) f(x,\alpha,t)f(y,\eta,t) \dd (y,\eta) \dd (x,\alpha)\\
=& \frac{1}{d}\intiaia  \Bigl(\gamma\nabla_x \xi(x,\alpha) \cdot \left(\phi_{xy\alpha}\odot(y-x)\right) + \mathcal{O}(\gamma^2)\Bigr) f(x,\alpha,t)f(y,\eta,t) \dd (y,\eta) \dd (x,\alpha).
\end{align*}
Note that we obtain the same equation as in \eqref{eq:derivation_withtaylor} except for the pre-factor $\frac{1}{d}$ on the right-hand side. We can thus follow the same steps as before to obtain the weak formulation
\begin{align*}
        \frac{\dd}{\dd t}\intia \xi(x,\alpha) f(x,\alpha,t) \dd (x,\alpha)
    =\frac{1}{d}\intiaia \nabla_x \xi(x,\alpha) \cdot (\phi_{xy\alpha}\odot (y-x)) f(x,\alpha,t)f(y,\eta,t) \dd (y,\eta) \dd (x,\alpha),
\end{align*}
and the strong formulation
\begin{align*}
    \frac{\partial}{\partial t} f(x,\alpha,t) = - \frac{1}{d}\nabla_x \cdot \left(\left(\intia  \phi_{xy\alpha}\odot (y-x) f(y,\eta,t)\dd (y,\eta)\right) f(x,\alpha,t)\right).
\end{align*} 
This shows that if the people talk only about one topic per interaction and the probability that they talk about a topic is the same for all topics, the dynamics are the same compared to the case where they talk about all topics at each conversation just the timescale is slower by the factor $\frac{1}{d}$.
\end{document}